\documentclass{sn-jnl}

\usepackage{graphicx}
\usepackage{amsmath} 
\usepackage{amssymb}

\usepackage{lineno}

\usepackage{url}
\usepackage{hyperref}

\newcommand{\true}{\mathsf {true}}
\newcommand{\false}{\mathsf {false}}

\newcommand{\C}{{\cal C}}

\newcommand{\fold}{\mathsf{fold}}

\newcommand{\suff}{\mathsf{suff}}
\newcommand{\pref}{\mathsf{pref}}

\newcommand{\vars}{\mathsf{vars}}

\newcommand{\keyword}[1]{\mathsf{#1}}
\newcommand{\SKIP}{\keyword{skip}}

\newcommand{\ASG}[2]{#1 := #2}
\newcommand{\SEQ}[2]{#1\,{;}\,#2}
\newcommand{\IFNZ}[3]{\keyword{if} \,(#1)\, #2\,\keyword{else}\,#3}

\newcommand{\WHILE}[2]{\keyword{while} \,(#1)\, #2}

\newcommand{\APP}[2]{\mathrm{app} (#1,#2)}

\newcommand{\CLO}[3]{\mathrm{clo} (#1,#2,#3)}

\newcommand{\aof}{\mathsf{atom\_or\_foldable}}

\usepackage{listings} 
\usepackage{courier}
\def\ll{[\![}
\def\rr{]\!]}

\def\anno#1{{\ooalign{\hfil\raise.07ex\hbox{\small{\rm #1}}\hfil \crcr\mathhexbox20D}}}
        
\theoremstyle{thmstyleone}\newtheorem{theorem}{Theorem}

\newtheorem{lemma}[theorem]{Lemma}

\theoremstyle{thmstyletwo}\newtheorem{example}{Example}\newtheorem{remark}{Remark}

\theoremstyle{thmstylethree}\newtheorem{definition}{Definition}

\title[Big-step and small-step Horn clause derivations applied to operational semantics]{Big-step and small-step Horn clause derivations applied to operational semantics$^\dagger$}

\author*[1,2]{\fnm{John P.} \sur{Gallagher}}\email{jpg@ruc.dk}

\author[2,3]{\fnm{Manuel} \sur{Hermenegildo}}

\author[2,3]{\fnm{Jos\'e} \sur{Morales}}

\author[2,4]{\fnm{Pedro} \sur{Lopez-Garcia}}

\author[3,2]{\fnm{Louis} \sur{Rustenholz}}

\affil*[1]{\orgname{Roskilde University}, \country{Denmark}}

\affil[2]{\orgname{IMDEA Software Institute},  \city{Madrid},  \country{Spain}}

\affil[3]{\orgname{Universidad Polit\'ecnica de Madrid (UPM)},  \city{Madrid},  \country{Spain}}

\affil[4]{\orgname{Spanish Council for Scientific Research (CSIC)},  \city{Madrid},  \country{Spain}}

\begin{document}
\abstract{
The concepts of big-step and small-step derivations are familiar from the operational
semantics of programming languages.  These concepts are applicable in the
more general setting of Horn clause derivations.  We prove equivalence between
big-step derivations and two versions of small-step 
derivations for Horn clauses.
By specialising interpreters for these derivation strategies, 
any set of Horn clauses can be transformed into a provably equivalent set of clauses 
that inherits the behaviour of a given (big- or small-step) Horn clause interpreter.
As a special case of this transformation, big-step semantics for any programming language,
expressed directly as Horn clauses, can be transformed into
equivalent small-step semantics.
Experiments with a variety of programming languages are reported.
}
\maketitle


\pagestyle{myheadings}

{ \renewcommand{\thefootnote}{$^\dagger$}
  \footnotetext{Partially funded by MICIU projects CEX2024-001471-M
    María de Maeztu and TED2021-132464B-I00 PRODIGY.
}
}

\section{Introduction}\label{intro}
In this paper we explore the connections between Horn clause derivations
and the formal operational semantics of programming languages.
Horn clauses have rich connections to a range of formal computational topics. 
First identified by the logician Alfred Horn in 1951 \cite{Horn_1951}, they
have been studied as a computational formalism from which
 the constraint logic programming (CLP) 
family of languages emerged \cite{Green69,Colmerauer,Kowalski74,Tarnlund77,KornerLBCDHMWDA22}.
Other areas studied through the medium of Horn clauses
include databases (Datalog), context-free languages (definite clause grammars)
and automated verification  (constrained Horn clauses) \cite{BjornerGMR15}.

A close relationship between Horn clauses and operational semantics
was  identified many years ago \cite{Kahn87}.
This paper extends that connection, shedding light on the relationship between 
different operational semantics
for the same language, especially between big-step and small-step semantics, allowing
automatic translation between the two forms.

\paragraph{Outline of the paper} 
In Section \ref{prelim} the notation and essential concepts for Horn clause derivations are presented.
Section \ref{Horn} contains operational semantics rules defining three Horn clause derivations;  
capturing both classical Horn clause derivations and also a novel ``contextual small-step" derivation
strategy.  We also show that the conjunctions (or ``stack") of atomic goals that arise in small-step 
derivations can be folded into a single atomic goal, using auxiliary ``suffix" clauses.

In Section \ref{opsem} we present a version of the contextual small-step derivations
applied to a particular class of Horn clauses, namely big-step operational semantics rules for
programming languages. The particular features of predicates and clauses in such rules
are exploited to refine the derivation, especially the auxiliary clauses allowing the
stack to be folded.
Section \ref{specsem} outlines the well-established method of program transformation using
interpreter specialisation (the first Futamura projection). The transformation is applied to big-step 
semantics rules, resulting in small-step rules.  The transformation is first applied to big-step semantics
of simple
imperative and functional languages, and then a larger example of the transformation of the
semantics of the reversible programming language Janus is outlined. 
Finally in Section \ref{concl} we review the relation to existing work, and outline future work.
Details of the experiments
and proofs of propositions are in appendices.

\section{Preliminaries}\label{prelim}

A first-order language is defined by $\langle \Sigma, \Pi,  \cal{V}, \cal{C} \rangle$, where $\Sigma$ and $\Pi$ are
non-empty sets of function symbols
and predicate symbols respectively, each having an assigned arity,
$\cal{V}$ is a countable set of variables, and $\cal{C}$ is a finite set of logic connectives and quantifiers.

A \emph{term} is a variable $v \in \cal{V}$, or an expression $f(t_1,\ldots,t_n)$ where $f \in \Sigma$,
$n \ge 0$ is the arity of $f$ 
and $t_1,\ldots,t_n$  are terms.  We use the letters $s,t,\ldots$ (with subscripts) to denote terms.

An \emph{atomic formula} or simply \emph{atom} is of the form $p(t_1,\ldots,t_n)$, where $p \in \Pi$
and $n \ge 0$ is the arity of $p$.  If $n=0$ in either terms or atoms the argument brackets are omitted.
We use the letters $A, B, C \ldots$ possibly with subscripts to
denote atoms.

We consider the logical connectives and quantifiers $\cal{C} = \{\vee, \wedge, \neg, \rightarrow, \forall, \exists\}$, 
with their usual
interpretations.

Quantifiers will mostly be implicit for the purposes of this paper, as all 
variables in formulas will be considered
to be universally quantified unless otherwise specified.

A \emph{literal} is either $A$ or $\neg A$ where $A$ is an atom.  $A$ and $\neg A$ are called \emph{positive}
and \emph{negative} literals respectively.

A \emph{clause} is a disjunction $L_1 \vee \cdots \vee L_k$ ($k \ge 0)$ where $L_i$ are literals.

A \emph{Horn clause} is a clause containing at most one positive literal.
A Horn clause $A_0 \vee \neg A_1 \vee \cdots \vee \neg A_k$ ($k \ge 0$)
is often written as the logically equivalent formula
$(A_1 \wedge \cdots \wedge A_k) \rightarrow A_0$.  Sometimes we will write this in the logic programming
fashion with the arrow reversed, $A_0 \leftarrow (A_1 \wedge \cdots \wedge A_k)$, or simply 
$A_0 \leftarrow A_1,\ldots,A_k$.  $A_0$ is called the \emph{head} of the clause and $A_1,\ldots,A_k$
is called the \emph{body} of the clause.  A Horn clause $\leftarrow A_1,\ldots,A_k$ (that is, with no head)
is called a \emph{goal}, and is equivalent to $\neg \exists (A_1 \wedge \cdots \cdots \wedge A_k)$.

A \emph{substitution} is a set of pairs $v/t$ where $v \in \cal{V}$ and $t$ is a term.
 If $\theta$ is a substitution and $S$ is any syntactic object then
 $S\theta$ denotes the \emph{application} of $\theta$ to $S$, obtained by simultaneously replacing
 every occurrence of $v$ in $S$ by $t$, where $v/t \in \theta$.  The \emph{empty substitution} is denoted $\epsilon$.
  
 Given two syntactic objects $S_1$ and $S_2$, $S_1$ is an 
 \emph{instance} of $S_2$ if there exists a 
 substitution $\phi$ such that $S_1 = S_2\phi$. 
 $S_1 \cong S_2$ holds if $S_1$ and $S_2$ are instances of each other; that is, they are identical 
 modulo renaming of variables.
  A \emph{unifier} of $S_1$ and $S_2$
 is a substitution such that $S_1\theta = S_2\theta$.
 A \emph{most general unifier} of $S_1$ and $S_2$ (or $\mathtt{mgu}(S_1,S_2)$ for short)
 is a unifier $\theta$,
 such that for any other unifier $\theta'$, $S_1\theta'$ is an instance of $S_1\theta$.
 The unification algorithm computes such a substitution, if it exists.  The $\mathtt{mgu}$ of
 two objects is
 unique, modulo renaming of variables.
 
 The composition of two substitutions $\theta$ and $\phi$ is the substitution denoted $\theta  \phi$, where
 for all syntactic objects $S$, $S(\theta  \phi) = S\phi\theta$ (that is, first apply $\phi$, then $\theta$).
 
 \subsection{Horn clause derivations}
 
Let $P$ be a set of Horn clauses. The \emph{resolution} inference rule for Horn clauses
is defined as follows.
Let $A_0 \leftarrow A_1,\ldots,A_k$ and $B_0 \leftarrow B_1,\ldots,B_m$ be two clauses in $P$.
The two clauses are assumed to have no variable in common - variables in clauses can always be 
renamed apart if necessary. 
Suppose $\mathtt{mgu}(A_0,B_j) = \theta$. Then a \emph{resolvent} of
the clauses is a clause 
\[
(B_0 \leftarrow B_1,\ldots, B_{j-1},A_1,\ldots,A_k,B_{j+1},\ldots,B_m)\theta.
\]
\noindent
In the following we consider a restricted form in which the second clause is a goal, 
that is, $B_0$ is absent, and $j$=$1$, that is, the first literal
of the goal is selected for resolution,

Let $P$ be a set of Horn clauses and $G_0$ a goal. A resolution \emph{derivation} for $P \cup \{G_0\}$
is a sequence of goals 
$G_0,G_1,\ldots$
where each $G_i$ ($i \ge 1$) is obtained by resolving $G_{i-1}$ with a clause from $P$.

A resolution \emph{refutation} of a goal $G_0$ in 
$P$ is a finite derivation in which the final goal is the empty clause. For
simplicity we will use the word ``proof" instead of refutation, but to be accurate, a refutation
of a goal $\leftarrow B_1,\ldots,B_k$
is a proof that $\exists(B_1,\ldots,B_k)$ follows from $P$.

This is historically called SLD-resolution \cite{AptE82,Lloyd,Stark89} (Selective Linear Definite clause resolution)
where the selection rule picks the leftmost literal.

\begin{theorem}[Soundness and completeness of SLD resolution]\label{SLD}
Let $P$ be a set of Horn clauses and let $A$ be an atom. Let $\theta$ and $\phi$ be 
substitutions.
Then $A\theta$ is a logical consequence of $P$ if and only if there is an SLD-resolution proof 
for $P \cup \{\leftarrow A\}$, $\phi$ is the composition of the substitutions arising from the
$\mathtt{mgu}$ operation at each step in the proof, and $A\theta$ is an instance of $A\phi$.
\end{theorem}

\subsection{Alternative formulation without explicit substitutions}

We write $A_0 \leftarrow A_1,\ldots,A_k \in [P]$ to mean that $A_0 \leftarrow A_1,\ldots,A_k$
is an instance of a clause in $P$.
Let $\leftarrow B_1,\ldots,B_m$ be a goal. Then $A_1,\ldots,A_k,B_2,\ldots,B_m$ is a resolvent
of the goal with (an instance of) the clause if $B_1 \leftarrow A_1,\ldots,A_k \in [P]$ (resolving on the leftmost
atom $B_1)$. 

The definition of a resolution derivation can then be expressed as follows:  
a resolution \emph{derivation} for $P \cup \{G_0\}$
is a sequence of goals 
$G_0,G_1,\ldots$
where each $G_i$ ($i \ge 1$) is a resolvent of $G_{i-1}$ with an instance of a clause from $P$.
The correctness theorem is restated as follows.

\begin{theorem}[Soundness and completeness of SLD resolution (without explicit substitutions)]\label{SLD2}
Let $P$ be a set of Horn clauses and let $A$ be an atom. Then $A\theta$ is a logical consequence of 
$P$ if and only if there is an SLD-resolution proof (without explicit substitutions)
for $P \cup \{\leftarrow A\theta\}$.
\end{theorem}
This formulation is similar to that used by St{\"a}rk \cite{Stark89} in the proof of soundness and 
completeness of SLD-resolution, except that St{\"a}rk formulated a derivation in terms of ``implication-trees"
(St{\"a}rk's term for AND-trees)
rather than derivation sequences.  It may seem odd that this formulation requires the ``answer" $\theta$
to be known before constructing a run.  In practice,
given an atom $A$, an instance $A\theta$ that is a logical consequence of a set of clauses $P$
is constructed during the derivation; $A$ is successively instantiated
at each step of the derivation.
We return in more detail to this style of proof in 
Section \ref{metaprog}.

In the following section we formulate operational semantics for Horn clause derivations
without explicit substitutions.

\section{Small-step and big-step Horn clause derivations}\label{Horn}

SLD-derivations can be structured in different ways.
We start by reformulating the notion of a derivation in the style of small-step operational semantics.

\subsection{SLD small-step Horn clause derivations}
A small-step semantics consists of three components:
\begin{enumerate}
\item
A set of \emph{configurations}, a non-empty subset of which is designated as \emph{terminal} configurations.
\item
A \emph{transition} relation $C \Rightarrow C'$, where $C,C'$ are configurations.
\item
The transitive closure of the transition relation $\Rightarrow\!\! *$. There is a \emph{run} from $C$ to $C'$ iff
$C \,{\Rightarrow}\!\! *\, C'$, and the run is \emph{successful} if $C'$ is a terminal configuration.
\end{enumerate}

For Horn clause derivations, a configuration is a goal $G$.
The goal $ \leftarrow \true$ is the terminal configuration.
For conciseness, we omit the $\leftarrow$ in goals; for example, we write $B_1,B_2$
instead of $\leftarrow B_1,B_2$.

The letter $R$ with subscripts represents a (possibly empty) conjunction of atoms;  thus $B,R$ represents the 
conjunction of which the atom $B$ is the first conjunct and $R$ is a conjunction. We can also write $R_1,R_2$ to represent the conjunction $A_1,\ldots,A_n,B_1,\ldots,B_m$, 
where $R_1 = A_1,\ldots,A_n$ and $R_2 = B_1,\ldots,B_m$, since the
conjunction connective is associative.  We also note that $(\true,R)$ is the same as $R$.

A run can be defined by the following rules 
where $R, R', R''$ are configurations and the arrow $\Rightarrow$ is any of the small-step transition arrows that
we present later ($\stackrel{1}{\Rightarrow}, \stackrel{2}{\Rightarrow}, \stackrel{3}{\Rightarrow}$).
\[
\begin{array}{ll}\label{rules-run}
\mathtt{[RUN0]}~~\dfrac{}
 {R \,{\Rightarrow}\!\! * \,R}~~~~~~&
 \mathtt{[RUN1]}~~\dfrac{R {\Rightarrow} \,R'~~~ R' {\Rightarrow}\!\! * \,R''}
 {R \,{\Rightarrow}\!\! * \, R''}
 \end{array}
\]

A transition for the SLD version of small-step derivations (given a set of Horn clauses $P$) is of the form
$G_0  \stackrel{1}{\Rightarrow} G_1$,
defined by the following small-step rule.
\[
\begin{array}{l}
\mathtt{[RES]}~~\dfrac{}
 {(B,R_1) \stackrel{1}{\Rightarrow} (R,R_1)}
 ~~~~\text{where } B \leftarrow R  \in [P]
 \end{array}
 \]

This, together with the definition of a run, is just a reformulation of an SLD-derivation 
(without explicit substitutions), so Theorem \ref{SLD2}
can be reformulated to state that $A\theta$ is a consequence of the given set of Horn clauses $P$
if and only if 
there exists a (successful) run
$A\theta  \,\stackrel{1}{\Rightarrow}\!\! * ~\true$  using the rules above.

 \subsection{Contextual small-step Horn clause derivations}

In the second version of the small-step derivation
a small step ends with a resolution step
using a clause with empty body, that is $A \leftarrow$, usually written as 
$A \leftarrow \true$.  Resolution with such a clause shortens the goal and represents finishing   
proving one conjunct and moving on to the next.

The transition relation $\stackrel{2}{\Rightarrow}$ is defined recursively using two rules for atoms, and two for conjunctions, shown in
Figure \ref{fig-ss-contextual}.

\begin{figure}
\[
\begin{array}{c}
\begin{array}{l}
\mathtt{[RES0]}~~\dfrac{}
 {B \stackrel{2}{\Rightarrow}\,\true}
 ~~~\text{where }B \leftarrow \true  \in [P]\\
\\
\mathtt{[RES1]}~~\dfrac{R  \stackrel{2}{\Rightarrow} \,R_1}
 {B \stackrel{2}{\Rightarrow} \,R_1}
 ~~~\text{where } B \leftarrow R  \in [P], R\neq \true\\
\end{array}\\
\\

\\
\begin{array}{ll}
\mathtt{[CONJ0]}~~ \dfrac{R_1 \stackrel{2}{\Rightarrow}  \true}
 {(R_1 \wedge R_2)\stackrel{2}{\Rightarrow} R_2}~~~~~~~&
 \mathtt{[CONJ1]}~~\dfrac{R_1 \stackrel{2}{\Rightarrow} R'_1}
 {(R_1 \wedge R_2) \stackrel{2}{\Rightarrow} \,(R'_1 \wedge R_2)}
  ~~~\text{where } R'_1\neq \true\\
 \end{array}\\
\end{array}
\]
\caption{Contextual small-step rules for Horn clauses}\label{fig-ss-contextual}
\end{figure}
The rules defining a run are $\mathtt{[RUN0]}$ and $\mathtt{[RUN1]}$ given above with $\Rightarrow\!\!\! *$
replaced by $\stackrel{2}{\Rightarrow\!\! *}$.
The form of $\mathtt{[CONJ0]}$ and
$\mathtt{[CONJ1]}$  is familiar as the pattern for 
``context rules" in small-step operational semantics.  Given a
compound structure (in this case a conjunction $(R_1 \wedge R_2)$), we 
make a small step on the first component of the
structure, $R_1$.  If $R_1$ is reduced to $\true$ in one small step, then move
on to $R_2$, using $\mathtt{[CONJ0]}$. Otherwise, let the result of a
small step on $R_1$ be $R'_1$, and continue with $(R'_1 \wedge R_2)$ using $\mathtt{[CONJ1]}$.

\begin{remark}\label{remark1}
For conjunctions in SLD derivations, the comma is used as the conjunction symbol
to indicate that a conjunction is regarded as a list of atoms and 
the derivation rule for $\stackrel{1}{\Rightarrow}$ assumes associativity of the comma symbol.
However, for contextual small-step derivations, and big-step derivations, conjunctions
are trees built using the $\wedge$
symbol;  and the rules for $\stackrel{2}{\Rightarrow}$ (Figure \ref{fig-ss-contextual}) and 
$\Downarrow$ (Section \ref{bigstep}) do not assume its associativity.  
\end{remark}

\begin{example}\label{ex1}
The following small example illustrates the difference between the two versions of small-step semantics.
Let $P$ be the set of (propositional) Horn clauses given below.

\[
\begin{array}{llll}
a \leftarrow b,c ~~~~~~~~~~~~~~~~	& d \leftarrow ~~~~~~~~~~~~~& e\leftarrow~~~~~~~~~~&g \leftarrow 	 \\

b \leftarrow d,e					& c \leftarrow f,g		& f \leftarrow&\\

\end{array}
\]
We construct a run starting from $a$.  
Using SLD small-step rules, the run is
\[
a \stackrel{1}{\Rightarrow} (b,c) \stackrel{1}{\Rightarrow} (d,e,c) \stackrel{1}{\Rightarrow} (e,c) \stackrel{1}{\Rightarrow} c \stackrel{1}{\Rightarrow}  (f,g) \stackrel{1}{\Rightarrow} g \stackrel{1}{\Rightarrow} \true
\]
Using the contextual small-step version, each transition ends with a resolution with a clause with empty body, so we get:
\[
a \stackrel{2}{\Rightarrow}  (e \wedge c) \stackrel{2}{\Rightarrow} c \stackrel{2}{\Rightarrow}  g \stackrel{2}{\Rightarrow} \true
\]

\end{example}

A run $R\stackrel{1}{\Rightarrow}\!\!*  ~R'$ is \emph{expanding} if no clause of the form
$B \leftarrow \true$ is used in the derivation.
We write
$R\stackrel{1}{\Rightarrow}_e\!\!*  ~R'$ for an expanding run.

\begin{lemma}\label{expanding-run}
For all expanding runs $R\stackrel{1}{\Rightarrow}_e\!\!*~R'$ where $R\neq \true$,
there exist $B,S,S'$ such that $R=(B,S)$ and $R'=(S',S)$. 
\end{lemma}

\begin{lemma}\label{opsem1-lemma}
For all conjunctions $R, R',$ and  $R''$, where $R\neq\true$,
$R\stackrel{1}{\Rightarrow}_e\!\!*  ~R'$  if and only if
$(R,R'') \stackrel{1}{\Rightarrow}_e \!\! * ~(R',R'')$, and both runs have the same height.

\end{lemma}

\begin{lemma}\label{opsem1}
There is a contextual small-step transition $R \stackrel{2}{\Rightarrow} R'$ if and only if  there is a run $R \stackrel{1}{\Rightarrow}_e\!\!*  ~(B,R' )\stackrel{1}{\Rightarrow} R'$, that is, an expanding run followed by a step using a clause $B\leftarrow \true$.
\end{lemma}

\begin{theorem}\label{opsem2}
For a given set of Horn clauses, $R \stackrel{1}{\Rightarrow}\!\!* ~\true$ 
if and only if 
$R \stackrel{2}{\Rightarrow}\!\!*  ~\true $.
\end{theorem}

 \subsection{Big-step Horn clause derivations}\label{bigstep}
 
 For Horn clause derivations, a big-step transition is of the form $A  \Downarrow \true$.  
 Note the downward transition
 arrow $\Downarrow$ contrasting with the small-step transition $\Rightarrow$.  There is no need for the notion
 of a ``run" since a big-step transition itself represents a complete derivation.
 The rules for $\Downarrow$ are as follows.
 \[
  \begin{array}{l}
 \mathtt{[BASE]}~~\dfrac{}
         {B \Downarrow  \true}
         ~~~\text{where }B \leftarrow \true \in [P] \\

 \\
\mathtt{[UNF]}~~ \dfrac{
R \Downarrow \true}
 {B \Downarrow \true}
 ~~~\text{where }B \leftarrow R \in [P], R\neq\true\\
  \\

\mathtt{[CONJ]}~~ \dfrac{
R \Downarrow  \true ~~~~~R' \Downarrow  \true}
 {(R \wedge R') \Downarrow \true}
 \end{array}\\
\]
 
The difference from the small-step semantics is the treatment of conjunctions. A successful
derivation for a conjunction is constructed from successful derivations for each of the conjuncts.

\begin{theorem}\label{opsem3}
For a given set of Horn clauses, $A \stackrel{1}{\Rightarrow}\!\!* ~  \true$ 
if and only if 
$A \Downarrow \true$.
\end{theorem}
\begin{proof}
Straightforward by induction on the height of the $\Downarrow$ derivation tree.
\end{proof}
Theorems \ref{opsem2} and \ref{opsem3} establish that all three operational semantics are equivalent with respect to successful
derivations.

\begin{figure}[t]
\begin{tabular}{c}
{\includegraphics[width=\textwidth]{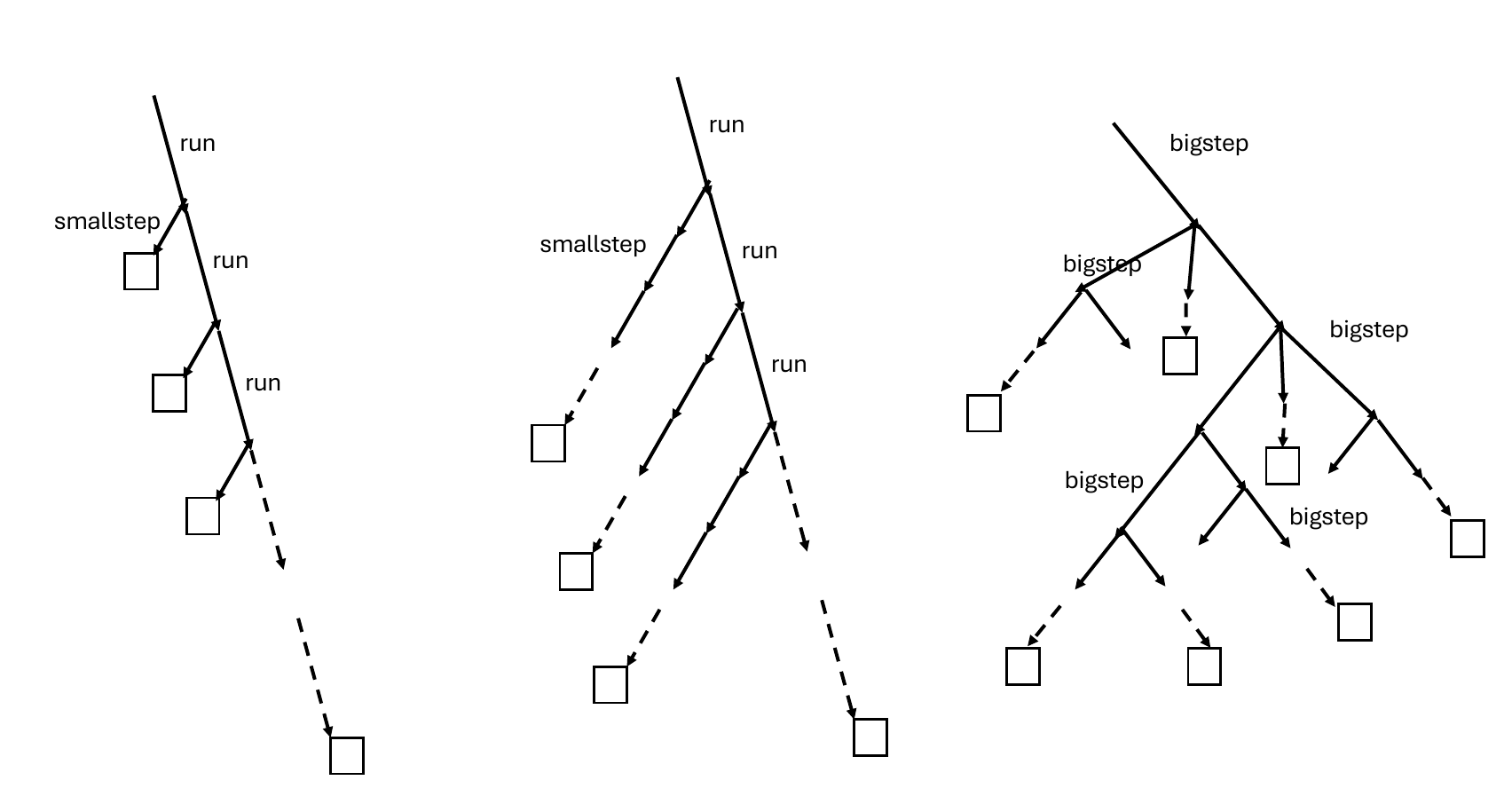}}\\
\end{tabular}
\caption{Derivation tree shape for SLD small-step  (left), contextual small.step (centre) and big-step (right)}\label{treeshape}
\end{figure}

\subsection{Folding the conjunctions in small-step derivations}\label{folding}

 The conjunctions that arise in small-step derivations starting from
 an atomic goal using 
$\stackrel{2}{\Rightarrow}$ have a structure due to the fact that the leftmost atom is always selected
for resolving with a clause head, and after deriving $\true$ from the leftmost atom, the next leftmost atom
is selected. Thus the conjunctions following a transition have the following form.
\[
((\cdots(R_1 \wedge R_2) \wedge \cdots R_{n-1})\wedge R_n)
\]
where each conjunction $R_i$ is an instance of a \emph{suffix} of a clause body.
In the following section, we show that the whole conjunction can be represented by a single atom, 
by \emph{folding} 
it using a set of auxiliary \emph{suffix rules}.

\subsubsection{Suffix rules}

Given a Horn clause of the following general form
\[
A_0 \leftarrow A_1,\cdots, A_n
\]
we assume the body to have a standard structure, so that $A_1, A_2, \ldots, A_n$
is the conjunction $A_1 \wedge (A_2 \wedge (\ldots \wedge A_n)\ldots)$.  (See also Remark \ref{remark1}).
We define the $i^{th}$ suffix  and $i^{th}$ prefix ($1 < i \le n$) of the rule as follows.
\[
\begin{array}{ll}
\suff(i) =& A_i \wedge ( \ldots \wedge A_n)\\
\pref(i) =& \{A_0,A_1,\ldots,A_{i-1}\}\\
\end{array}
\]
For each $i$ ($1 < i \le n$) we construct two \emph{suffix rules}, defining a predicate $\mathtt{stack}$
(since the conjunction can be seen as the stack of calls awaiting selection).  
\[
\begin{array}{ll}
\text{Type 1 suffix rule: }&\mathtt{stack}(f_i(\bar{w})) \leftarrow \suff(i)\\
& ~~~~~\text{where } \bar{w}= \vars(\pref(i)) \cap \vars(\suff(i))
\\
\\
\text{Type 2 suffix rule: }&\mathtt{stack}(g_i(\bar{w} \cup \{z\})) \leftarrow \mathtt{stack}(z) \wedge \suff(i)\\
&~~~~~\text{where } \bar{w}= \vars(\pref(i)) \cap \vars(\suff(i))\\
&~~~~~\text{and } z \text{ is a fresh variable}
\end{array}
\]
$f_i$, $g_i$ are function symbols occurring only in the respective rules.
Note that $\pref(1)$ and $\suff(1)$ are not used
for forming suffix rules.

The purpose of the suffix rules is to allow the conjunctions arising from a 
$\stackrel{2}{\Rightarrow}$  transition to be folded into a single $\mathtt{stack}$ atom.
The following example introduces the idea of folding;  a precise definition of
folding follows the example.

\begin{example}\label{rev-ex} 

Consider the Horn clauses defining ``naive" reverse.
 \begin{lstlisting}[basicstyle=\small\ttfamily]
rev([],[]).
rev([X|Xs],Ys) :-
	rev(Xs,Ws), app(Ws,[X],Ys).

app([],Ys,Ys).
app([X|Xs],Ys,[X|Zs]) :-
	app(Xs,Ys,Zs).
 \end{lstlisting}
 The suffix clauses are generated from the second clause using $\suff(2) = \mathtt{app(Ws,[X],Xs)}$.
  \begin{lstlisting}[basicstyle=\small\ttfamily]
stack(f_2_1(X,Ys,Ws)) :-
	app(Ws,[X],Ys).
stack(g_2_1(X,Ys,Ws,NewX)) :-
	stack(NewX),
	app(Ws,[X],Ys).
 \end{lstlisting}
There is a transition 
\[
\mathtt{rev([a,b,c],X)} \stackrel{2}{\Rightarrow}
(\mathtt{app([],[c],X2)} \wedge \mathtt{app(X2,[b],X1)}) \wedge \mathtt{app(X1,[a],X)}
\]
The atom $\mathtt{app([],[c],X2)}$ is an instance of the body of the
first suffix clause above, and can be folded to $\mathtt{stack(f\_{2\_1}(c,X2,[]))}$.  
This gives the conjunction:
\[
(\mathtt{stack(f\_{2\_1}(c,X2,[]))} \wedge \mathtt{app(X2,[b],X1)}) \wedge \mathtt{app(X1,[a],X)}
\]
The conjunction $\mathtt{stack(f\_{2\_1}(c,X2,[]))} \wedge \mathtt{app(X2,[b],X1)}$
is an instance of the body of the second suffix clause, and can be folded.  The
same clause can then be used for folding again, 
yielding the single atom $\mathtt{stack(g\_{2\_1}(a,X,X1,g\_{2\_1}(b,X1,X2,f\_{2\_1}(c,X2,[])))}$.

   \end{example}
 The argument of the $\mathtt{stack}$ predicate encodes the conjunction of calls to atoms as a nested term
 that carries the information contained in the conjunction.  
 Local variables in the suffixes (that is, variables in $\suff(i) \setminus \pref(i)$) are not in the folded atom.
 The identical conjunction (up to renaming of such local variables)
 can
 be reconstructed by unfolding the atom using the $\mathtt{stack}$ clauses.

\paragraph{Foldable conjunctions} 
Let $P_{\mathrm{suff}}$ be the set of suffix clauses
derived from a set of Horn clauses $P$. 
Each clause in $P_{\mathtt{suff}}$ is assigned the
index of the clause in $P$ from which it is derived. There is no need to provide 
separate indices to different suffixes from the same clause, as their length is sufficient to
distinguish them. We annotate conjunctions and write $R^k$ when there exists a $k$ such that $R$ is the body of a suffix clause with index $k$.

A \emph{foldable conjunction} is of the form
\[
((\cdots(R^{k_1}_1 \wedge R_2)^{k_2} \wedge \cdots R_{n-1})^{k_{n-1}}\wedge R_n)^{k_n}
\]
where $R_2,\ldots,R_n$ are instances of bodies of suffix clauses of type 1, and $R_1$ is
an instance of
the body of a suffix clause of type 1 or type 2.
The position of  annotations determines which suffix clauses are involved in 
a foldable conjunction. For instance $(A^{k_1} \wedge B^{k_2})$ is a different foldable
conjunction than $(A \wedge B)^k$.  In the first, $A$ and $B$ are instances of the bodies of 
suffix clauses $k_1$ and $k_2$ respectively, while in the second, $A \wedge B$ is
an instance of
the body of suffix clause $k$.  

\begin{remark}\label{unannotated}
If an
unannotated conjunction is stated to be ``foldable", then we mean that annotations could be added to make it
a foldable conjunction.  More precisely, if $R$ is ``foldable" then either $R$ is a suffix  
of a clause,
or $R = R_1 \wedge R_2$, where $R_2$ is a suffix of a clause
and the unannotated conjunct $R_1$ is ``foldable".
This slight abuse of notation is used in Lemma \ref{foldable} and also in Theorem \ref{fig3-fig4}.
Note that a ``foldable" unannotated conjunction might be annotated in different ways; consider
for example the unannotated conjunction $A \wedge B$  with the annotations just discussed.
\end{remark}

Revisiting Example \ref{rev-ex}, and adding annotations, the conjunction
$((\mathtt{app([],[c],X2)}^1 \wedge \mathtt{app(X2,[b],X1)})^1 \wedge \mathtt{app(X1,[a],X)})^1$
contains three instances of the body of suffix clause body 1 (of type 1).
After folding the first instance, we obtain
\[
((\mathtt{stack(f\_{2\_1}(c,X2,[]))} \wedge \mathtt{app(X2,[b],X1)})^1\wedge \mathtt{app(X1,[a],X)})^1
\]
in which $(\mathtt{stack(f\_{2\_1}(c,X2,[]))} \wedge \mathtt{app(X2,[b],X1)})^1$
is an instance of a type 2 suffix clause body with index 1.  Whether the index $k$ indicates a suffix clause of
type 1 or 2 is decided by whether a $\mathtt{stack}$ atom occurs as the first atom of the conjunct
component (in which case it is of type 2).

\begin{definition}[Folding]\label{unfold}

Let $P_{\mathtt{suff}}$ be the set of suffix rules derived from $P$.  
The function $\fold$ is applied to foldable conjunctions or $\true$ (which indicates the empty 
foldable conjunction).

\[
\begin{array}{c}
\fold(S) = 
\begin{cases}
A, \text{ if } S=R^k \text { and }A \leftarrow R \in_k [P_{\mathtt{suff}}]\\
\fold(\fold(R^{k_1}_1) \wedge R)^{k}) \text{ if } S =  R^{k_1}_1 \wedge R^{k}\\
\true \text{ if } S = \true
\end{cases}
\end{array}
\]
\end{definition}
It can be seen that $\fold(S)$, where $S$ is a foldable conjunction or $\true$,
is a $\mathtt{stack}$ atom or $\true$.  
The following property relating folding and suffix clauses also follows directly from the definition.
\begin{enumerate}
\item
$B \leftarrow R \in_k P_{\mathrm{suff}}$ is a suffix clause of type 1, if and only if $B=\fold(R^k)$.
\item
$B \leftarrow B' \wedge R \in_k P_{\mathrm{suff}}$ is a suffix clause of type 2 if and only if 
$B' = \fold(R^{k_1}_1)$ for some $R^{k_1}_1$ and $B=\fold(R^{k_1}_1 \wedge R^k)$.
\end{enumerate}

\begin{lemma}\label{folding2}
Let $S_1, S_2$ be foldable conjunctions or $\true$.  Then
\[
S_1 \cong S_2 \text{ if and only if }\fold(S_1) \cong \fold(S_2). 
\]
Note that the equalities are modulo renaming of variables.
\end{lemma}

Lemma \ref{folding2} shows that the $\fold$ function gives a unique representation of a
conjunction that arises in a contextual small-step derivation, as a $\mathtt{stack}$ atom.  We will now 
systematically replace the conjunctions in a derivation by their corresponding foldings.
First, we show that foldability is preserved in $\stackrel{2}{\Rightarrow}$ derivations.

\begin{lemma}\label{foldable}
Let $B$ be an atom and $R$ be a
foldable conjunction (note Remark \ref{unannotated}).  Then
 \begin{enumerate}
 \item
If $B\stackrel{2}{\Rightarrow} R_1$ then $R_1$ is a foldable conjunction or $\true$.
 \item
If $R\stackrel{2}{\Rightarrow} R_1$ then $R_1$ is a foldable conjunction or $\true$.
 \end{enumerate}

\end{lemma}

\begin{figure}
\[
\begin{array}{c}
\begin{array}{l}
\mathtt{[RES0]}~~\dfrac{}
 {B \stackrel{2}{\Rightarrow}\,\true}
 ~~~\text{where }B \leftarrow \true  \in [P]\\
\\
\mathtt{[RES1}_1]~~\dfrac{B'  \stackrel{2}{\Rightarrow} \,R}
 {B \stackrel{2}{\Rightarrow} \,R}
 ~~~\text{where } B \leftarrow B'  \in [P], B'\neq \true\\
 \\
 \mathtt{[RES1}_2]~~\dfrac{B_1   \stackrel{2}{\Rightarrow} \,\true}
 {B \stackrel{2}{\Rightarrow} \,R}
 ~~~\text{where } B \leftarrow B_1 \wedge R  \in [P]\\
 \\
 \mathtt{[RES1}_3]~~\dfrac{B_1   \stackrel{2}{\Rightarrow} \,R_1}
 {B \stackrel{2}{\Rightarrow} \,(R_1 \wedge R)}
 ~~~\text{where } B \leftarrow B_1 \wedge R  \in [P], R_1 \neq \true\\

\end{array}
\\
\\
\begin{array}{ll}
\mathtt{[CONJ0]}~~ \dfrac{R_1 \stackrel{2}{\Rightarrow}  \true}
 {(R_1 \wedge R)\stackrel{2}{\Rightarrow} R}~~~~~~~&
 \mathtt{[CONJ1]}~~\dfrac{R_1 \stackrel{2}{\Rightarrow} R_2}
 {(R_1 \wedge R) \stackrel{2}{\Rightarrow} \,(R_2 \wedge R)}
  ~~~(R_2\neq \true)\\
 \end{array}
 \end{array}
 \]
\caption{Transition rules for contextual small-step derivations, adapted from Figure \ref{fig-ss-contextual}}\label{ss_v2_norm}
\end{figure}

\subsection{Contextual small-step derivations with folding}
Figure \ref{ss_v2_norm} shows the small-step transition rules for $\stackrel{2}{\Rightarrow}$, modified from 
Figure \ref{fig-ss-contextual} by some straightforward equivalence-preserving unfoldings.
Rule $\mathtt{[RES1}]$ is unfolded by splitting into cases depending on whether the interpreted 
clause has one or more premises, and then by unfolding the premise in the second case using
$\mathtt{[CONJ0]}$ and $\mathtt{[CONJ1]}$. 
 The resulting rules are  $\mathtt{[RES1}_1]$,
$\mathtt{[RES1}_2]$ and $\mathtt{[RES1}_3]$ respectively.

\begin{figure}
\[
\begin{array}{c}
\begin{array}{l}
\mathtt{[UNF0]}~~\dfrac{}
 {B \stackrel{3}{\Rightarrow}\,\true}
 ~~~\text{where }B \leftarrow \true  \in [P]\\
\\
\mathtt{[UNF1]}~~\dfrac{B'  \stackrel{3}{\Rightarrow} \,\fold(R^k)}
 {B \stackrel{3}{\Rightarrow} \,\fold(R^k)}
 ~~~\text{where } B \leftarrow B'  \in [P\cup P_{\mathrm{suff}}], B'\neq \true\\
 \\
 \mathtt{[FOLD0]}~~\dfrac{B_1  \stackrel{3}{\Rightarrow} \,\true}
 {B \stackrel{3}{\Rightarrow} \,\fold(R^k)}
 ~~~\text{where } B \leftarrow B_1 \wedge R  \in_k [P \cup P_{\mathrm{suff}}]\\
\\
\mathtt{[FOLD1]}~~\dfrac{B_1  \stackrel{3}{\Rightarrow} \,\fold(R^{k_1}_1)}
 {B \stackrel{3}{\Rightarrow} \,\fold((\fold(R^{k_1}_1) \wedge R)^k)}
 ~~~\text{where } B \leftarrow B_1 \wedge R  \in_k [P \cup P_{\mathrm{suff}}], R_1 \neq \true\\
\\
\end{array}
 \end{array}
 \]
\caption{Contextual small-step rules from Figure \ref{ss_v2_norm}, 
applying $\fold$ to foldable conjunctions.}\label{ss_v2_norm_fold}
\end{figure}

\paragraph{Incorporating folding into the transition rules}
We now incorporate
the $\fold$ function into the semantic rules of Figure \ref{ss_v2_norm}.
To do so, we restrict transitions $R_1 \stackrel{2}\Rightarrow R_2$ such
that $R_1$ is either an atom or a foldable conjunction, and
$R_2$  is either $\true$ or a foldable conjunction.
We will later show (Theorem \ref{fig3-fig4}) that this restriction is justified.
Due to this restriction, the transition arrow is renamed
$\stackrel{3}{\Rightarrow}$.

Firstly, we drop rules $\mathtt{[CONJ0]}$ and $\mathtt{[CONJ1]}$, since 
folded conjunctions are $\mathtt{stack}$ atoms and thus there is no need for
rules handling conjunctions.  In rules 
$\mathtt{[RES1}_1]$, $\mathtt{[RES1}_2]$ and $\mathtt{[RES1}_3]$, 
(which are renamed to $\mathtt{[UNF1]}$, $\mathtt{[FOLD0]}$ and $\mathtt{[FOLD1]}$) the
set of Horn clauses $P$ is replaced by $P \cup P_{\mathrm{suff}}$ so that
the suffix rules can be used to handle $\mathtt{stack}$ atoms. 

Secondly, we add clause indices to the rules. For example, we write
$B \leftarrow B_1 \wedge R  \in_k [P\cup P_{\mathrm{suff}}]$ in the condition of 
rule $\mathtt{[FOLD1]}$, 
meaning that $k$ is the index of the selected
clause.  
These indices are then used to annotate the conjunctions appearing in the rules.
In this case, wherever the conjunct $R$ appears in a rule, it is given the index $k$.

Then each
conjunction of the form $R_1 \wedge R$, where $R$ is a suffix of clause $k$, is replaced 
by $\fold((\fold(R^{k_1}_1) \wedge R)^k)$.
By assumption, $R_1$ and $R$ are foldable conjunctions.
$R$ occurring on the right of a transition is replaced by
$\fold(R^k)$, for some $k$.
(If $R$ evaluates to $\true$ we assume its index is 0).
We assign all occurrences of  a conjunction $R$ in a rule the same index.
The resulting rules are shown in Figure 
\ref{ss_v2_norm_fold}.  

Following this, we eliminate the $\fold$ function from Figure \ref{ss_v2_norm_fold}. 
Wherever an expression of form $\fold(R^k)$ appears more than
once in a rule,
it represents the same atom, say $B'$. This follows from Lemma \ref{folding2}. 
Replace such expressions by atoms throughout the rule.
After doing this, whenever an expression $\fold(B' \wedge R^k)$ or $\fold(R^k)$ appears, 
replace the
expression by an atom $B''$, and add a condition of the form $B'' \leftarrow B' \wedge R \in_k [P_{\mathtt{suff}}]$
or $B'' \leftarrow R \in_k [P_{\mathtt{suff}}]$ respectively. This change is justified by the
definition of $\fold$ and Lemma \ref{folding2}.
The resulting set of four rules is
shown in Figure \ref{ss_v3}.  

\begin{figure}
\[
\begin{array}{ll}
~[\mathtt{UNF0}]~~~&\dfrac{}{ B \stackrel{3}{\Rightarrow}  \true}
~~~~\text{where }B \leftarrow \true \in [P]\\
\\
~[\mathtt{UNF1}]~~~&\dfrac{B_1 \stackrel{3}{\Rightarrow}  B_2}
{B \stackrel{3}{\Rightarrow}  B_2}
~~~\text{where }B \leftarrow B_1  \in [P \cup P_{\mathtt{suff}}] \\
 \\
~[\mathtt{FOLD0}]~~~&\dfrac{B_1 \stackrel{3}{\Rightarrow} \true}
{B\stackrel{3}{\Rightarrow}  B'}
~~~\text{where }B \leftarrow B_1,R_2  \in_k [P \cup P_{\mathtt{suff}}] \text{ and }B' \leftarrow R_2 \in_k [P_{\mathtt{suff}}]   \\
\\
~[\mathtt{FOLD1}]~~~&\dfrac{B_1 \stackrel{3}{\Rightarrow}  B_2}
{B\stackrel{3}{\Rightarrow} B'}
~~\text{where }B \leftarrow B_1,R   \in_k [P \cup P_{\mathtt{suff}}] \text{ and }
B' \leftarrow B_2,R \in_k [P_{\mathtt{suff}}]  \\
\end{array}
 \]
 \caption{Contextual small-step rules from Figure \ref{ss_v2_norm_fold}, 
 after eliminating the $\fold$ function}\label{ss_v3}
 \end{figure}

\begin{example}\label{rev-suff}
Consider the naive reverse Horn clauses again and their suffix rules (Example \ref{rev-ex}).
With the rules of Figure \ref{ss_v3} (and the rules for a run), we obtain the following run 
from $\mathtt{rev([a,b,c],X)}$.
\[
\begin{array}{l}
\mathtt{rev([a,b,c],X)}\\
 \stackrel{3}{\Rightarrow} \mathtt{stack(g\_2\_1(a,X,X1,g\_2\_1(b,X1,X2,f\_2\_1(c,X2,[]))))}\\
\stackrel{3}{\Rightarrow} \mathtt{stack(g\_2\_1(a,X,X1,f\_2\_1(b,X1,[c])))}\\
\stackrel{3}{\Rightarrow} \mathtt{stack(f\_2\_1(a,X,[c,b]))}\\
\stackrel{3}{\Rightarrow} \true\\
\end{array}
\]
The substitutions $\mathtt{X=[c,b,a], X1=[b,a], X2=[a]}$ are applied to the variables in the run.
Within each small step is a sub-derivation.  For instance:
\[
\begin{array}{lll}
&\mathtt{stack(g\_2\_1(a,X,[c,b],f\_2\_1(b,[c,b],[c])))} \stackrel{3}{\Rightarrow} \mathtt{stack(f\_2\_1(a,X,[c,b]))}
&\text{ using }[\mathtt{FOLD0}]\\
\text{since}&\mathtt{stack(f\_2\_1(b,[c,b],[c]))} \stackrel{3}{\Rightarrow} \true&\text{ using }[\mathtt{UNF1}]\\
\text{since}&\mathtt{app([c],[b],[c,b]))} \stackrel{3}{\Rightarrow} \true&\text{ using }[\mathtt{UNF1}]\\
\text{since}&\mathtt{app([],[b],[b])} \stackrel{3}{\Rightarrow} \true&\text{ using }[\mathtt{UNF0}]\\
\end{array}
\]
\end{example}
As already noted, the transition rules for $\stackrel{3}{\Rightarrow}$ represent a restriction
compared to $\stackrel{2}{\Rightarrow}$ since they apply to foldable conjunctions, whereas  
there can be a transition $R \stackrel{2}{\Rightarrow} R'$
in which the conjunctions $R$ and $R'$ are arbitrary conjunctions.  Lemma \ref{foldable}
and the following theorem 
establish that
$\stackrel{2}{\Rightarrow}$ derivations from atoms and foldable conjunctions
are preserved in $\stackrel{3}{\Rightarrow}$ derivations.

 \begin{theorem}\label{fig3-fig4}
Let $P$ be a set of Horn clauses, $B$ be an atom in the signature of $P$ and $R,R_1$ be 
foldable conjunctions (note Remark \ref{unannotated}).
 Then
 \begin{enumerate}
 \item
$B\stackrel{2}{\Rightarrow} R_1$ if and only if $\exists k_1.B\stackrel{3}{\Rightarrow} \fold(R_1^{k_1})$.
 \item
 $R\stackrel{2}{\Rightarrow} R_1$ if and only if $\exists k_1.k.\fold(R^k)\stackrel{3}{\Rightarrow} \fold(R_1^{k_1})$.
 \end{enumerate}

 \end{theorem}

In particular,  the following holds by
induction on the height of the run.

\begin{theorem}
 Let $P$ be a set of Horn clauses and
 let $B$ be an atom in the signature of $P$. Then there is a run $B\stackrel{3}{\Rightarrow}\!\!* \,\true$
 if and only if there is a run $B \stackrel{2}{\Rightarrow}\!\!* \,\true$.
 \end{theorem}

The usefulness of the contextual small-step transition system might not be obvious so far, when looking at
cases such as Example \ref{rev-ex}. In the next section we consider a particular class of
Horn clauses, namely sets of rules for big-step operational semantics of programming languages.
In this setting,  contextual small-step Horn clause derivations provide a framework for
transforming big-step semantics to small-step semantics. 

\section{Horn clauses for programming language semantics}\label{opsem}

We are now going to apply small-step derivations to a particular class of 
Horn clauses, namely, programming language
operational semantics.  
The rules of operational semantics, whether big-step or small-step, have premises and a conclusion,  written as follows.
\[
\dfrac{\mathtt{premises}} 
{\mathtt{conclusion}} 
~~~~\text{if } \mathtt{conditions}
\]

With suitable encoding of the arguments of transitions, this is an implication in first-order predicate logic:
\[
\mathtt{premises} \wedge \mathtt{conditions} \rightarrow \mathtt{conclusion}.
\]

The conclusion is an atomic formula
(a big- or small-step transition with predicate symbol such as 
$\Rightarrow$, $\Rightarrow\!\!\!*$ and $\Downarrow$). 
Assuming that the premises and conditions are conjunctions,
it is a Horn clause.

The close connection between semantics rules and Horn clauses, and 
hence to the logic programming language Prolog, 
was observed by Kahn and his co-workers and exploited in
the Typol tool \cite{Despeyroux1984-short,Kahn87}. In short, a set of operational semantics rules, either big-step or
small-step, is a set of Horn clauses, which can be given a procedural reading using one of the 
derivation strategies discussed above, just as any set of Horn clauses.

We will show that we can give a small-step interpretation to a set of big-step rules, or
a big-step interpretation to a set of small-step rules.  This will form the basis of transformations of
big-step semantics to small-step, and vice versa.

\subsection{Big-step rules}\label{bigstep-rules}

A big-step transition for programming language semantics has the form $\langle s, \sigma \rangle \Downarrow v$, 
meaning that the language 
construct $s$ is completely evaluated in environment $\sigma$, resulting in value $v$. 
In some languages, the final value is also an environment;  the environment itself can be a more or less complex object comprising
stores, function definitions, status flags and so on.  

To simplify the presentation we assume that a language definition has only
one big-step transition relation, even though in many languages more than one is used.  
For example, there may be one big-step transition for
evaluating expressions and another for executing statements, and each may depend on the other. 
In our formulation we use only one relation and assume that they are distinguished by their syntax argument.
A big-step semantics for a language consists of a set of rules of the following form.

\[
\dfrac{c_1 ~~~\langle s_1, \sigma_1\rangle \Downarrow v_1 ~~~~\cdots~~~ c_n~~~\langle s_n, \sigma_n\rangle \Downarrow v_n ~~~c_{n+1}}
{\langle s_0, \sigma_0\rangle \Downarrow v_0} 
\]
where $c_1,\ldots,c_{n+1}$ are some side conditions and operations interspersed among the transition premises.

\subsubsection{Suffix rules for big-step clauses}

Naturally, we could form suffix clauses defining a predicate $\mathtt{stack}$ as shown earlier,
since the construction applied to any Horn clauses.
However, we will exploit the particular form of big-step transition rules to
refine the rules in the following ways.
\begin{enumerate}

\item
The arguments of the stack predicate will mirror those of the big-step transition, name $s$, $\sigma$ and $v$.

\item
Going further, we will also reuse the predicate name $\Downarrow$ instead of the  
predicate $\mathtt{stack}$. 
In principle we could use any predicate name, provided that the 
function symbol in the head of a suffix clause is unique to that clause,
to avoid conflict with other clauses when re-unfolding a folded stack.
\item
However, by choosing the predicate name $\Downarrow$, we are also able to merge the
two kinds of suffix rule into a single kind.  We elaborate the justification below.

\end{enumerate}

For ($1 \le i \le n$) we define the $i^{th}$ suffix  and $i^{th}$ prefix of the big-step rule as follows.  
These are essentially
the same as the prefix and suffix defined for Horn clauses previously, but we also include the $c_i$
components of the rule, which are not executed in small steps.
\[
\begin{array}{ll}
\suff(i) =& \langle s'_i, \sigma'_i\rangle \Downarrow v_i ~~~c_{i+1}~~~\cdots~~~ c_n~~~\langle s_n, \sigma_n\rangle \Downarrow v_n ~~~c_{n+1}\\
\pref(i) =& \langle s_0, \sigma_0\rangle \Downarrow v_0~~~c_1 ~~\langle s_1, \sigma_1\rangle\Downarrow v_1~~~\cdots~~~ c_{i-1}~~~\langle s_{i-1}, \sigma_{i-1}\rangle \Downarrow v_i ~~~c_{i}\\
\end{array}
\]
Note:
\begin{itemize}

\item
In $\suff(i)$, $s_i, \sigma_i$ from the rule are replaced by $ s'_i, \sigma'_i$,
where $s'_i,\sigma'_i$ are fresh variables not
occurring elsewhere in the rule.  The variable
$v_i$ is not renamed, since we will show in Section \ref{value-arg} that the value argument in transitions is unchanged
during derivations.
Replacing $\langle s_i, \sigma_i\rangle \Downarrow v_i$
by $\langle s'_i, \sigma'_i\rangle \Downarrow v_i$ corresponds to replacing $A_i$ by $\mathtt{stack}(z)$
in the second suffix rule for Horn clauses, but now we use the predicate $\Downarrow$
instead of $\mathtt{stack}$, and the fresh variables $ s'_i, \sigma'_i$ corresponding to the
fresh variable $z$ of $\mathtt{stack}(z)$.

\item
$\pref(1) = \langle s_0, \sigma_0\rangle \Downarrow v_0~~~c_1$ since $\langle s_i, \sigma_i\rangle\Downarrow v_i$ is not
included in $\pref(i)$.

\end{itemize}

For each $1 \le i \le n$  the \emph{suffix rule} is

\[
\dfrac{\suff(i)}
{\langle f_i(\bar{s}), \bar{\sigma}\rangle \Downarrow v_0} 
\]
where $f_i$ is a function symbol not used in any rule in $P$, and $\bar{s}$, $\bar{\sigma}$ contains the 
variables needed to evaluate the suffix, given the prefix. More precisely: let $\bar{u} = \vars(\pref(i))$
and $\bar{v} = \vars(\suff(i))$. Let $\bar{w}= \bar{u} \cap \bar{v}$.
The set $\bar{w}$ is split (somewhat arbitrarily) into disjoint sets,``syntax" variables and ``environment" 
variables, say $\bar{s}$ and $\bar{\sigma}$ respectively, which appear in the rule conclusion.  This is only for
readability of the rule; the exact split is not critical. 
In short, in the suffix rule, the rule conclusion contains variables that occur both in $\pref(i)$ and $\suff(i)$.  

We now show that these suffix rules subsume suffix rules of the first and second kind
introduced
 for arbitrary Horn clauses in Section \ref{Horn}.
For $1 < i \le n$ the above rule corresponds to the first kind of suffix rule except that the first 
premise is generalised by introducing
fresh variables.  Therefore any suffix foldable without the fresh variables is also foldable
by the rule with fresh variables.  Now let us consider suffix rules of the second form, which would
be as follows, for $1 < i \le n$.
\[
\dfrac{\langle s',\sigma'\rangle \Downarrow v' ~~c_{i+1}~~\suff(i)}
{\langle f_i(\bar{s}), \bar{\sigma}\rangle \Downarrow v_0} 
\]
where $s',\sigma',v'$ are fresh variables.  Then the premise $\langle s',\sigma'\rangle \Downarrow v' ~~c_{i+1}~~\suff(i)$ is an instance of $\suff(i-1)$.  This remain the case when $v'$ is replaced by $v_{i-1}$ (see Section 
\ref{value-arg}).  Therefore, we do not need to generate suffix rules of the second kind, since any 
conjunction foldable by the $i^{th}$ suffix rule of the second kind is foldable by the 
$i-1^{th}$ suffix rule of the first kind.  This also explains why we generate suffix rules for 
$1 \le i \le n$ instead of $1 < i \le n$.

\begin{example}\label{normal-ex} 
The following rule is from the big-step semantics for call-by-value semantics of the $\lambda$-calculus, giving the meaning
of $\mathrm{app}(e_1,e_2)$, i.e. ``apply $e_1$ to $e_2$".
\[
\dfrac{\langle e_1, \rho \rangle \Downarrow \CLO{x}{e}{\rho'} ~~~\langle e_2, \rho \rangle \Downarrow v_2~~~
\rho'[x/v_2] = \rho'' ~~~\langle e, \rho'' \rangle \Downarrow v} 
{\langle \APP{e_1}{e_2}, \rho \rangle \Downarrow v}
\]

This rule has three suffix rules.
\[
\begin{array}{l}

\dfrac{\langle e'_1, \rho''' \rangle \Downarrow \CLO{x}{e}{\rho'} ~~~\langle e_2, \rho \rangle \Downarrow v_2~~~
\rho'[x/v_2] = \rho'' ~~~\langle e, \rho'' \rangle \Downarrow v} 
{\langle \mathrm{app1}({e'_1},{e_2}), (\rho''',\rho) \rangle \Downarrow v}
\\
\\
\dfrac{\langle e'_2, \rho''' \rangle \Downarrow v_2 ~~~\rho'[x/v_2] = \rho'' ~~~\langle e, \rho'' \rangle \Downarrow v}
{\langle \mathrm{app2}({x},{e},{e'_2}), (\rho''',\rho') \rangle \Downarrow v}
\\
\\

\dfrac{\langle e', \rho' \rangle \Downarrow v}
{\langle \mathrm{app3}(e'),{\rho'} \rangle \Downarrow v}
 \\
 \end{array}
 \]
 The syntax constructors $\mathrm{app1}$, $\mathrm{app2}$ and $\mathrm{app3}$ are the new constructor
 symbols corresponding to $f_i$ in the definition of a suffix rule.  Indexed variants of the syntax constructor 
 appearing on the original rule are chosen for readability.

\end{example}
\subsection{A further refinement -- elimination of the value argument}\label{value-arg}

A big-step transition $C \Downarrow v$ transforms a configuration $C$ to a value $v$.  
A small-step transition based on Horn clauses transforms one call $C_1 \Downarrow v_1$
to another call $C_2 \Downarrow v_2$, except for the
final transition in a run, which takes $C \Downarrow v_1$ to $\true$.

By contrast, small-step transitions for programming languages are of the form $C_1 \Rightarrow C_2$
except the final transition in a run, which is $C\Rightarrow v$.  In other words, the final value
is not mentioned in a derivation except in the final transition.

We now show how to modify the Horn clause derivation to
produce small-step transitions of the expected form, namely
$\langle s_1,\rho_1\rangle \stackrel{3}{\Rightarrow} 
\langle s_2,\rho_2\rangle $
by eliminating the value arguments.

 \begin{definition}[Final value condition]\label{finalvalue}
 Let 
 \[
\dfrac{\langle s_1, \sigma_1\rangle \Downarrow v_1, \ldots, \langle s_n, \sigma_n\rangle \Downarrow v_n}
{\langle s_0, \sigma_0\rangle \Downarrow v_0} 
~~~~\mathrm{if }~ c
\]
be a big-step transition rule.    Condition $c$ is included in the premises and can appear anywhere among
the premises. Then the rule satisfies the \emph{final value condition} if $v_0$ does not occur in the premises
except in the final premise, and if the final premise is $\langle s_n, \sigma_n\rangle \Downarrow v_n$,
then $v_0=v_n$.
 \end{definition}
We have not encountered any semantics rules that do not satisfy the final value condition, but if there were
one, we could modify it by replacing any occurrences of $v_0$ in non-final premises by 
a fresh variable, say $v'$, and then add $v_0=v'$ as the final premise.  Suffix rules derived from
clauses satisfying the condition also satisfy it.

\begin{lemma}\label{final-value-lemma}
Let $P$ be a set of big-step rules satisfying the final value condition (Definition \ref{finalvalue}).
Let $(\langle s_1,\rho_1\rangle \Downarrow v_1) \stackrel{3}{\Rightarrow} 
(\langle s_2,\rho_2\rangle \Downarrow v_2)$ be derived 
using the rules in Figure \ref{ss_v3}.
Then $v_1 = v_2$.  
\end{lemma}

As a result of Lemma \ref{final-value-lemma}, the rules of Figure \ref{ss_v3} can be amended to eliminate the argument
$v$ from  atoms $\langle s,\rho\rangle \Downarrow v$.  Each transition of the form
$B \stackrel{3}{\Rightarrow} B'$, where $B=\langle s,\rho\rangle \Downarrow v$ and 
$B'=\langle s',\rho'\rangle \Downarrow v'$
can be replaced by $\langle s,\rho\rangle \stackrel{3}{\Rightarrow} \langle s',\rho' \rangle$,
defined as 
\[
(\langle s,\rho\rangle \stackrel{3}{\Rightarrow} \langle s',\rho' \rangle \equiv \exists v,v' .(\langle s,\rho\rangle \Downarrow v)  
\stackrel{3}{\Rightarrow} (\langle s',\rho'\rangle \Downarrow v') \wedge v=v'
\]
Furthermore a transition of the form
$B \stackrel{3}{\Rightarrow} \true$, where $B=\langle s,\rho\rangle
\Downarrow v$ can be replaced by
$\langle s,\rho\rangle \stackrel{3}{\Rightarrow} v$.
The modified rules for contextual small-step Horn derivations of big-step semantics rules, 
with conjunctions folded using suffix rules, and the removal of the value argument,
are shown in Figure \ref{ss_v3_a}.

\begin{figure}
\[
\begin{array}{lll}
~[\mathtt{UNF0}]~&\dfrac{}{ \langle s,\rho \rangle \stackrel{3}{\Rightarrow}  v}
&\text{where }(\langle s,\rho \rangle \Downarrow v) \leftarrow \true \in [P]\\
\\
~[\mathtt{UNF1}]~&\dfrac{\langle s_1,\rho_1 \rangle \stackrel{3}{\Rightarrow}  \langle s_2,\rho_2 \rangle}
{\langle s,\rho \rangle \stackrel{3}{\Rightarrow}  \langle s_2,\rho_2 \rangle}
&\text{where }(\langle s,\rho \rangle \Downarrow v) \leftarrow (\langle s_1,\rho_1 \rangle \Downarrow v)  \in [P] \\
 \\
~[\mathtt{FOLD0}]~&\dfrac{\langle s_1,\rho_1 \rangle \stackrel{3}{\Rightarrow} v_1}
{\langle s,\rho \rangle \stackrel{3}{\Rightarrow}  \langle s',\rho' \rangle}
&\text{where }(\langle s,\rho \rangle \Downarrow v) \leftarrow (\langle s_1,\rho_1 \rangle \Downarrow v_1),R_2  \in_k [P \cup P_{\mathtt{suff}}]\\
&&\text{ and }(\langle s',\rho' \rangle \Downarrow v) \leftarrow R_2 \in_k [P_{\mathtt{suff}}]   \\
\\
~[\mathtt{FOLD1}]~&\dfrac{\langle s_1,\rho_1 \rangle \stackrel{3}{\Rightarrow}  \langle s_3,\rho_3 \rangle}
{\langle s,\rho \rangle \stackrel{3}{\Rightarrow}  \langle s',\rho' \rangle}
&\text{where }(\langle s,\rho \rangle \Downarrow v) \leftarrow (\langle s_1,\rho_1 \rangle \Downarrow v),R_2  \in_k [P \cup P_{\mathtt{suff}}]\\
&&\text{ and }
(\langle s',\rho' \rangle \Downarrow v) \leftarrow (\langle s_3,\rho_3 \rangle \Downarrow v),R_2\in_k [P_{\mathtt{suff}}]  \\
\end{array}
 \]
 \caption{Contextual small-step rules for big-step semantics rules: folded conjunctions, suffix rules and value argument removed}\label{ss_v3_a}
 \end{figure}

\subsection{Optional refinements}\label{optional}
In some cases, a suffix rule is not needed.  Consider the big-step rule for executing a statement
sequence $s_1;s_2$ in an imperative language.
\[
\dfrac{\langle s_1,\rho\rangle \Downarrow \rho_1~~~~\langle s_2,\rho_1\rangle \Downarrow \rho_2}
{\langle s_1;s_2,\rho\rangle \Downarrow \rho_2}
\]
Using $\suff(1)$  we obtain a suffix rule.
\[
\dfrac{\langle s'_1,\rho'\rangle \Downarrow \rho_1~~~~\langle s_2,\rho_1\rangle \Downarrow \rho_2}
{\langle f(s'_1,s_2),\rho'\rangle \Downarrow \rho_2}
\]
The premise of the suffix rule is a variant of the premise of the original rule, and so the original rule
could be used for folding rather than the suffix rule.  The suffix rule in this case can just be
omitted, with the advantage that the extra syntax constructor $f$ is avoided.
This optimisation can be performed whenever the renamed terms $s_1$ and $\rho_1$ in $\suff(1)$
do not appear elsewhere in the premises.

Secondly, the suffix rule for $\suff(n)$ can be omitted, whenever  the $c_{n+1}$ premise is empty.
This requires splitting rule $[\mathtt{FOLD0}]$ into two cases depending on whether $R_2$ is
a single atom or not.  This results in the following two rules.
\[
\begin{array}{lll}
~[\mathtt{FOLD0}_1]~&\dfrac{\langle s_1,\rho_1 \rangle \stackrel{3}{\Rightarrow} v_1}
{\langle s,\rho \rangle \stackrel{3}{\Rightarrow}  \langle s_2,\rho_2\rangle}
&\text{where }\\
&&\langle s,\rho \rangle \Downarrow v \leftarrow (\langle s_1,\rho_1 \rangle \Downarrow v_1,
 \langle s_2,\rho_2 \rangle \Downarrow v) \in_k [P \cup P_{\mathtt{suff}}]\\
\\
~[\mathtt{FOLD0}_2]~&\dfrac{\langle s_1,\rho_1 \rangle \stackrel{3}{\Rightarrow} v_1}
{\langle s,\rho \rangle \stackrel{3}{\Rightarrow}  \langle s',\rho' \rangle}
&\text{where }\langle s,\rho \rangle \Downarrow v \leftarrow (\langle s_1,\rho_1 \rangle \Downarrow v_1,R_2)  \in_k [P \cup P_{\mathtt{suff}}]\\
&&\text{ and }\langle s',\rho' \rangle \Downarrow v \leftarrow R_2 \in_k [P_{\mathtt{suff}}], \vert R_2 \vert > 1  \\
\end{array}
 \]
In $[\mathtt{FOLD0}_1]$, the suffix $\langle s_2,\rho_2 \rangle \Downarrow v$ is used itself to
make the next configuration $\langle s_2,\rho_2\rangle$, rather than folding it using the suffix clause.
This can be done only because we reuse the predicate $\Downarrow$ as the
$\mathtt{stack}$ predicate.
\begin{example}\label{no-suffix}
Consider the suffix rules constructed in Example \ref{normal-ex}.  The third suffix, namely
\[
\dfrac{\langle e', \rho' \rangle \Downarrow v}
{\langle \mathrm{app3}(e'),{\rho'} \rangle \Downarrow v}
\]
can be omitted if we use the modified rules above.
\end{example}

This completes the presentation of contextual small-step semantics for Horn clauses,
with refinements for the special case where the Horn clauses are big-step semantics
rules.
We emphasise that the final set of rules presented, namely those in Figure \ref{ss_v3_a},
with the optional modification of rule $[\mathtt{FOLD0}]$ just discussed, 
are the same rules as for arbitrary Horn clauses shown in Figure \ref{ss_v3},
except for syntactic refinements exploiting 
the form of big-step rules.

\section{Horn clause interpreters and their specialisation}\label{specsem}
The operational semantics rules for Horn clause derivations developed in Section \ref{opsem} 
can be used as
interpreters for Horn clauses.
In this section, Horn clauses are themselves
used as the meta-language in which the interpreters are written.  

\subsection{A note on meta-programming techniques}\label{metaprog}

Let us call the interpreter the \emph{meta}-program and the set of Horn clauses being 
interpreted the \emph{object} program. Every syntactic object in the object program has a unique 
representation as a variable-free term in the
meta-program. An object program variable such as $x$ is represented by a term such as $\mathtt{'\$VAR'('x')}$.
The object program is represented in the meta-program 
as a variable-free term, such as
a list of clauses.  

In order to perform unification on object language terms, we adopt a technique described 
by Hill and Gallagher \cite{gallagher:pepm93,Hill-Gallagher-Handbook}.  Rather than compute an $\mathtt{mgu}$ of two terms, we compute the
\emph{most general instance} of the terms.
Let $t_1,t_2$ be (representations of) object language terms that we wish to unify.
Being variable-free terms in the meta-language, they are not directly unifiable in the meta-language, 
unless they happen to be
identical.  The meta-language relation $\mathtt{instance}(u,v)$,  is true if $v$ is an instance of $u$.
When given $u$ as input, it yields the most general instance $v$ rather than enumerating all possible instances. 

For example,  let $t$ be the object language term $\mathtt{f('\$VAR'('x'))}$.
$\mathtt{instance}(t,y)$
gives the answer substitution $y= \mathtt{f(}z\mathtt{)}$, where $z$ is a \emph{meta-language} variable.
Implicitly, $z$ ranges over all object-language terms.
Thus, to unify $t_1$ and $t_2$, the meta-program calls
$$\mathtt{instance}(t_1,x),  \mathtt{instance}(t_2,y), x=y.$$
This effectively asks whether there is a common instance of $t_1$ and $t_2$, and due to the way in which
the $\mathtt{instance}$ predicate is implemented, the answer returned is the most general instance, if
it exists.  The implementation of the $\mathtt{instance}$ predicate is purely logical;  the original Horn clause
code can be found in  \cite{Hill-Gallagher-Handbook}, Section 3.3, and our interpreters in this section 
contain similar code.

The correctness theorem states that $A$ is a logical consequence of $P$ if and only
if there is a derivation $A  \stackrel{1}{\Rightarrow}\!\!* ~ \true$. In
this style of meta-programming, we start with some atom $A'$ which is
 successively instantiated  to $A$ during the derivation;  the final instance at the end of a derivation
run is the instance that is proved. 

In this way operations such as unification, substitutions and renaming of object language terms
are ``reflected" in the meta-language in a purely logical way, when we use one of the available
implementations of Prolog as the meta-language\footnote{Our interpreters use 
only pure Horn clauses, avoiding non-logical features
of Prolog.}.

\begin{figure}
\begin{lstlisting}[basicstyle=\small\ttfamily]
  run(true,_).
  run(G,P) :-
  	G \== true,
	sstep(G,G1,P),
	run(G1,P).
		
  sstep(atom(A),true,P) :-
	rule(K,A,R,P),
	evalCondition(K,R,true,P).
  sstep(atom(A),R2,P) :-
	rule(K,A,R,P),
	evalCondition(K,R,R1,P),
	R1\==true,
	sstep(R1,R2,P).
  sstep(and(K,A,B),B1,P) :-	 
	sstep(A,true,P),
	evalCondition(K,[B],B1,P).
  sstep(and(K,A,B),and(K,A1,B),P) :-
	sstep(A,A1,P),
	A1\==true.
	
  rule(K,A,Bs,P) :-
	member(clause(Cl,K),P),		
	instanceOf(Cl,(A:-Bs)).
\end{lstlisting}
\caption{Contextual small-step interpreter for Horn clause derived from Figure \ref{fig-ss-contextual}.}\label{fig:ss_1_2}

\end{figure}

Figure \ref{fig:ss_1_2} shows the Horn clauses (in Prolog notation) defining runs and transitions for  
contextual small-step Horn clause derivations.  
This can be run as a Prolog program to interpret Horn clauses.

The predicate $\mathtt{rule(}K,A,B,P\mathtt{)}$ occurring in Figure \ref{fig:ss_1_2} means that $A \leftarrow B$
is an instance of the $K^{th}$ clause in the set of Horn clauses $P$.
The rules correspond directly to the respective semantics rules, using the 
meta-programming technique just discussed, apart from the fact that
the rules in the interpreter also allow for the evaluation of the premises
$c_1,\ldots,c_{n+1}$
in the general form of big-step rules presented in Section \ref{opsem}, that is, 
``built-ins" such as arithmetic
operations and other predicates that are evaluated directly rather than via the interpreter.  
These are handled by $\mathtt{evalCondition}$
in the interpreter.
Naturally, these could also be added to the semantics rules in Section \ref{opsem} but
we preferred to keep those rules conceptually simpler.

\subsection{Transformation by interpreter specialisation}

The idea of program specialisation, sometimes called partial evaluation or mixed computation, 
originated in the 1960s and 1970s
\cite{BeckmanHOS76,Ershov77,Futamura,Lombardi67}.  Its relation to compilation
was discovered by Futamura \cite{Futamura}.
We summarise here the approach to program transformation by specialisation
of interpreters \cite{Bruynooghe-DeSchreye-Krekels,Gallagher-86,GiacobazziJM12,Gluck94,GluckJ94,Jones04,NysS18,Turchin85}.

Let $\mathbf{P}$ be a program (a set of Horn clauses in our
applications). The function implemented by $\mathbf{P}$ is denoted $\mathbf{\ll P \rr}$, 
using the notational conventions of Jones \emph{et al.} \cite{Jones-Gomard-Sestoft}. The domain of
 $\mathbf{\ll P \rr}$
 is assumed to consist of pairs $\mathbf{(s,d)}$.  
$\mathbf{\ll P \rr} ~\mathbf{(s,d)}$ denotes the application of $\mathbf{\ll P \rr}$ to $\mathbf{(s,d)}$.

A \emph{specialiser} $\mathbf{Sp}$ is a program 
that is applied to $\mathbf{P}$ and one of its
inputs $\mathbf{s}$ (the ``statically known" argument).  The result is a specialised version of $\mathbf{P}$,
say $\mathbf{P_s}$, which, when applied to the remaining argument $\mathbf{d}$,
returns the same as $\mathbf{\ll P \rr} ~\mathbf{(s,d)}$.
\[
\begin{array}{lll}
\mathbf{\ll Sp \rr~(P,s)}  &=& \mathbf{P_s}\\
\mathbf{\ll P_s \rr~d} &=& \mathbf{\ll P \rr~(s,d)}
\end{array}
\]
Suppose $\mathbf{I}$ is an interpreter for some language.   
It takes as arguments a program $\mathbf{p}$ and its initial environment $\mathbf{e}$.
Then instantiating the equations above we have $\mathbf{\ll Sp \rr~(I,p)} = \mathbf{I_p}$ and 
$\mathbf{\ll I_p \rr~e} = \mathbf{\ll I \rr~(p,e)}$.  The program $\mathbf{I_p}$ inherits the behaviour
of the interpreter $\mathbf{I}$, including its complexity.  As Jones \cite{Jones04} pointed out, 
specialisation of interpreters can give non-linear speed-up, when comparing $\mathbf{I_p}$ with $\mathbf{p}$.
$\mathbf{I_p}$ also inherits the structure of its computations from $\mathbf{I}$.

\subsection{Specialisation of the Horn clause interpreters}\label{spec-int}

We have interpreters for Horn clauses
for the two versions of small-step derivations and big-step derivations, respectively.  By specialising them with
respect to a given set of Horn clauses $\mathbf{p}$ we obtain, different versions of $\mathbf{p}$. 
The three versions are equivalent, by Theorems \ref{opsem1},\ref{opsem2} and \ref{opsem3}.

\begin{example}\label{revspec}
Consider again the clauses and suffix clauses for naive reverse from Examples \ref{rev-ex} and \ref{rev-suff}.
Specialising an interpreter derived from contextual small-step rules of Figure \ref{ss_v3} we obtain the clauses in Figure \ref{fig:spec_rev_contextual}.  There are seven specialised versions of the \texttt{smallStep}
predicate from the interpreter, which represents the transition relation $\stackrel{2}{\Rightarrow}$.  
These correspond to seven small-step transition rules.
When the call \texttt{run\_\_1(rev([a,b,c],X)} is executed, the same sequence of atoms in calls to \texttt{run\_\_1}
is generated as shown in Example \ref{rev-suff}.
\end{example}

The transformed program in Figure \ref{fig:spec_rev_contextual} could not itself be considered useful: 
efficiency is not improved, 
and the specialised clauses in Figure \ref{fig:spec_rev_contextual} are less readable than the original ones. 
Nonetheless the transformation shows that in principle any set of Horn clauses can be
transformed to a set that behaves in a small-step manner.  When applied to big-step rules for a given
programming language
the result corresponds exactly to small-step operational semantics rules for the same language.

\begin{figure}
\begin{lstlisting}[basicstyle=\small\ttfamily]
run__1(true) :-
   true.
run__1(A) :-
   smallStep__2(A,B),
   run__1(B).
smallStep__2(rev([],[]),true) :-
   true.
smallStep__2(app([],A,A),true) :-
   true.
smallStep__2(stack(f_2_2(A,B,C)),D) :-
   smallStep__2(app(A,[B],C),D).
smallStep__2(app([A|B],C,[A|D]),E) :-
   smallStep__2(app(B,C,D),E).
smallStep__2(stack(g_2_2(A,B,C,D)),stack(f_2_2(A,B,C))) :-
   smallStep__2(stack(D),true).
smallStep__2(stack(g_2_2(A,B,C,D)),stack(g_2_2(A,B,C,E))) :-
   smallStep__2(stack(D),stack(E)).
smallStep__2(rev([A|B],C),stack(f_2_2(D,A,C))) :-
   smallStep__2(rev(B,D),true).
smallStep__2(rev([A|B],C),stack(g_2_2(D,A,C,E))) :-
   smallStep__2(rev(B,D),stack(E)).
\end{lstlisting}
\caption{Contextual small-step version of Example \ref{rev-ex}}\label{fig:spec_rev_contextual}

\end{figure}

\section{Transforming big-step to small-step operational semantics}\label{big2small}
By specialising a Horn clause interpreter derived directly from the rules in Figure \ref{ss_v3_a}
with respect to a set of big-step rules, we transform big-step semantics rules to
small-step semantics rules.    The code for the interpreter is shown in Appendix \ref{secA3},

\begin{example}\label{while}
Figures \ref{while-big}, \ref{while-small-Horn} and \ref{while-small-latex}
show the result for a small imperative language.
Note that the rule for executing $s_1;s_2$ (represented $\mathtt{seq(S_1,S_2)}$) is used instead of a 
suffix rule, as explained in Section \ref{optional}.
In this example, the evaluation of expressions is not carried out in small steps, though it could have been.
In a later example (the Janus example) the transformation to small-step rules includes both
the execution of statements and the evaluation of expressions.  Configurations in the
interpreter are named either $\mathtt{val(V)}$ or $\mathtt{conf(C)}$ where these stand for values and 
configurations respectively (see Appendix \ref{secA3}).

\begin{figure}
\[
\begin{array}{l|l}
 \dfrac{ } 
{\langle \ASG{x}{e}, \sigma \rangle \Downarrow \sigma[x/v]}
 ~~~\text{if } V (e, \sigma) = v~
&~~~
\dfrac{\langle s_1, \sigma \rangle \Downarrow \sigma' } 
{\langle \IFNZ{b}{s_1}{s_2}, \sigma \rangle \Downarrow \sigma'}
 ~~~~~~~\text{if } V(b, \sigma )= \true\\
 \\
\dfrac{\langle s_1, \sigma \rangle \Downarrow \sigma' ~~~~\langle s_2, \sigma' \rangle \Downarrow \sigma''} 
{\langle \SEQ{s_1}{s_2}, \sigma \rangle \Downarrow \sigma''}~~~
&~~~
\dfrac{\langle s_2, \sigma \rangle \Downarrow \sigma' } 
{\langle \IFNZ{b}{s_1}{s_2}, \sigma \rangle \Downarrow \sigma'}
 ~~~~~~~\text{if } V(b ,\sigma) = \false\\
 \\
 \dfrac{} 
{\langle \SKIP, \sigma \rangle \Downarrow \sigma}
 ~
&~~~
\dfrac{\langle \IFNZ{b}{s;\WHILE{b}{s}}{\SKIP}, \sigma \rangle \Downarrow \sigma'} 
{\langle \WHILE{b}{s}, \sigma \rangle \Downarrow \sigma'}
\\ 
\end{array}
\]

\caption{Big-step rules for a simple imperative language}\label{while-big}
\end{figure}

\begin{figure}
\begin{tabular}{l}
\begin{lstlisting}[basicstyle=\scriptsize\ttfamily]
smallStep__2((skip,A),val(A)) :-
   true.
smallStep__2((asg(var(A),B),C),val(D)) :-
   eval__3(B,C,E,F),
   eval__4(A,F,E,D).
smallStep__2((ifthenelse(A,B,C),D),E) :-
   eval__3(A,D,F,1),
   smallStep__2((B,F),E).
smallStep__2((ifthenelse(A,B,C),D),E) :-
   eval__3(A,D,F,0),
   smallStep__2((C,F),E).
smallStep__2((while(A,B),C),D) :-
   smallStep__2((ifthenelse(A,seq(B,while(A,B)),skip),C),D).
smallStep__2((seq(A,B),C),conf((B,D))) :-
   smallStep__2((A,C),val(D)).
smallStep__2((seq(A,B),C),conf((seq(D,B),E))) :-
   smallStep__2((A,C),conf((D,E))).
\end{lstlisting}
\end{tabular}

\caption{Small-step clauses from the specialised interpreter.}\label{while-small-Horn}
\end{figure}

\begin{figure}
\[
\begin{array}{l|l}
\dfrac{} 
{\langle \IFNZ{b}{s_1}{s_2}, \sigma \rangle \Rightarrow \langle s_1,\sigma'\rangle}
 ~\text{if } V(b, \sigma )= \true
 &~\dfrac{ } 
{\langle \ASG{x}{e}, \sigma \rangle \Rightarrow \sigma[x/v]}
 ~~~\text{if } V (e, \sigma) = v~~~\\
\\
 \dfrac{} 
{\langle \IFNZ{b}{s_1}{s_2}, \sigma \rangle \Rightarrow \langle s_2,\sigma'\rangle}
 ~\text{if } V(b, \sigma )= \false
&~
\dfrac{} 
{\langle \SKIP, \sigma \rangle \Rightarrow \sigma}
\\
\\
\dfrac{\langle s_1, \sigma \rangle \Rightarrow \sigma' } 
{\langle \SEQ{s_1}{s_2}, \sigma \rangle \Rightarrow \langle s_2 ,\sigma'\rangle}~~~
&~
 \dfrac{\langle s_1, \sigma \rangle \Rightarrow \langle s_1', \sigma'\rangle} 
{\langle \SEQ{s_1}{s_2}, \sigma \rangle \Rightarrow \langle \SEQ{s_1'}{s_2}, \sigma'\rangle}~~~
 ~
\\
&~\\
\dfrac{} 
{\langle \WHILE{b}{s}, \sigma \rangle \Rightarrow \langle \IFNZ{b}{s;\WHILE{b}{s}}{\SKIP}, \sigma \rangle }~~
&

 \end{array}
\]
\caption{Small-step rules from Figure \ref{while-small-Horn} rewritten in standard form.}\label{while-small-latex}
\end{figure}
\end{example}

\begin{example}\label{mini-ml}
The Mini-ML language big-step semantics was presented by Kahn \cite{Kahn87} in the original introduction to
natural semantics (big-step semantics).
A notable feature of the rules is that the definition of the $\mathrm{letrec}$ construct for recursive
functions uses an environment containing circular terms:  the closure of a lambda expression for a 
recursive function includes
an environment which itself contains a link to the same function definition.  This construct can be 
handled using Horn clauses with equality over rational trees. The relation $\stackrel{\mathrm{val-of}}{\vdash}$ 
is for retrieving the value of an identifier from the environment.
In an application $(E_1 E_2)$, $E_1$ may be either a function defined by a $\lambda$-expression, or
a predefined function $\mathrm{Op}$.  In the latter case an evaluation relation $\stackrel{\mathrm{eval}}{\vdash}$
is used.  Figure \ref{miniml-small-Horn} contains the  small-step rules produced by 
specialising the interpreter based on Figure \ref{ss_v3_a}.

\begin{figure}
\[
\begin{array}{lll}
\dfrac{} 
{\rho \vdash \mathrm{number~N} \Downarrow \mathsf{N}}~~~~~~&
\dfrac{} 
{\rho \vdash \mathrm{true} \Downarrow \true}~~~~~~&
\dfrac{} 
{\rho \vdash \mathrm{false} \Downarrow \false}\\
\\
\end{array}
\]
\[
\begin{array}{ll}
\dfrac{} 
{\rho \vdash \lambda P.E \Downarrow \ll \lambda P.E, \rho\rr}~~~&
\dfrac{\rho \stackrel{\mathrm{val-of}}{\vdash} \mathrm{ident} ~I \mapsto \alpha}
{\rho \vdash \mathrm{ident} ~I \Downarrow \alpha}~~~
\\
 \end{array}
\]
\\
\[
\begin{array}{ll}
\dfrac{\rho \vdash E_1  \Downarrow \true ~~~\rho \vdash E_2  \Downarrow \alpha} 
{\rho \vdash \mathrm{if} ~E_1~ \mathrm{then}~E_2~\mathrm{else} ~E_3 \Downarrow \alpha}~~~&
\dfrac{\rho \vdash E_1  \Downarrow \false ~~~\rho \vdash E_3  \Downarrow \alpha} 
{\rho \vdash \mathrm{if} ~E_1~ \mathrm{then}~E_2~\mathrm{else} ~E_3 \Downarrow \alpha}~~~~~~
\\
 \end{array}
\]
\\
\[
\begin{array}{l}
\dfrac{\rho \vdash E_1  \Downarrow \alpha ~~~\rho \vdash E_2  \Downarrow \beta} 
{\rho \vdash (E_1,E_2) \Downarrow (\alpha,\beta)}\\
 \end{array}
\]
\\
\[
\begin{array}{l}
\dfrac{\rho \vdash E_1  \Downarrow  \mathrm{opaque}~\mathrm{Op}~~~\rho \vdash E_2  \Downarrow \alpha
~~~\stackrel{\mathrm{eval}}{\vdash} \mathrm{Op},\alpha \Downarrow \beta}
{\rho \vdash E_1~E_2 \Downarrow \beta}\\
 \end{array}
\]
\\
\[
\begin{array}{l}
\dfrac{\rho \vdash E_1  \Downarrow  \ll \lambda P.E, \rho_1\rr~~~\rho \vdash E_2  \Downarrow \alpha
~~~\rho_1[P \mapsto \alpha] \vdash \beta}
{\rho \vdash E_1~E_2 \Downarrow \beta}\\
 \end{array}
\]
\\
\[
\begin{array}{ll}
\dfrac{\rho \vdash E_2  \Downarrow  \alpha~~~~~\rho[P \mapsto \alpha] \vdash E_1 \Downarrow \beta}
{\rho \vdash \mathrm{let}~P=E_2 ~\mathrm{in}~E_1 \Downarrow \beta}~~~&
\dfrac{\rho[P \mapsto \alpha] \vdash E_2  \Downarrow  \alpha~~~~~
\rho[P \mapsto \alpha] \vdash E_1 \Downarrow \beta}
{\rho \vdash \mathrm{letrec}~P=E_2 ~\mathrm{in}~E_1 \Downarrow \beta}\\
 \end{array}
\]

\caption{Big-step rules for Mini-ML, as presented by Kahn \cite{Kahn87}.}\label{while-small-latex}
\end{figure}

\begin{figure}
\begin{tabular}{l}
\begin{lstlisting}[basicstyle=\footnotesize\ttfamily]
smallStep__2((number(A),B),val(int(A))) :-
   true.
smallStep__2((true,A),val(true)) :-
   true.
smallStep__2((false,A),val(false)) :-
   true.
smallStep__2((lambda(A,B),C),val(closure(A,B,C))) :-
   true.
smallStep__2((id(A),B),val(C)) :-
   eval__3(B,A,C).
smallStep__2((if(A,B,C),D),conf((B,D))) :-
   smallStep__2((A,D),val(true)).
smallStep__2((if(A,B,C),D),conf((if_8_1(B,E),D,F))) :-
   smallStep__2((A,D),conf((E,F))).
smallStep__2((if(A,B,C),D),conf((C,D))) :-
   smallStep__2((A,D),val(false)).
smallStep__2((if(A,B,C),D),conf((if_9_1(C,E),D,F))) :-
   smallStep__2((A,D),conf((E,F))).
smallStep__2((mlpair(A,B),C),conf((mlpair_10_2(D,B),C))) :-
   smallStep__2((A,C),val(D)).
smallStep__2((mlpair(A,B),C),conf((mlpair_10_1(B,D),C,E))) :-
   smallStep__2((A,C),conf((D,E))).
smallStep__2((apply(A,B),C),conf((apply_11_2(D,E,F,B),C))) :-
   smallStep__2((A,C),val(closure(D,E,F))).
smallStep__2((apply(A,B),C),conf((apply_11_1(B,D),C,E))) :-
   smallStep__2((A,C),conf((D,E))).
smallStep__2((apply(A,B),C),conf((apply_12_2(D,B),C))) :-
   smallStep__2((A,C),val(opaque(D))).
smallStep__2((apply(A,B),C),conf((apply_12_1(B,D),C,E))) :-
   smallStep__2((A,C),conf((D,E))).
smallStep__2((let(A,B,C),D),conf((C,[(A,E)|D]))) :-
   smallStep__2((B,D),val(E)).
smallStep__2((let(A,B,C),D),conf((let_13_1(A,C,E),D,F))) :-
   smallStep__2((B,D),conf((E,F))).
smallStep__2((letrec(A,B,C),D),conf((C,[(A,E)|D]))) :-
   smallStep__2((B,[(A,E)|D]),val(E)).
smallStep__2((letrec(A,B,C),D),conf((letrec_14_1(E,C,F),[(A,E)|D],G))) :-
   smallStep__2((B,[(A,E)|D]),conf((F,G))).
smallStep__2((if_8_1(A,B),C,D),conf((A,C))) :-
   smallStep__2((B,D),val(true)).
smallStep__2((if_8_1(A,B),C,D),conf((if_8_1(A,E),C,F))) :-
   smallStep__2((B,D),conf((E,F))).
smallStep__2((if_9_1(A,B),C,D),conf((A,C))) :-
   smallStep__2((B,D),val(false)).
smallStep__2((if_9_1(A,B),C,D),conf((if_9_1(A,E),C,F))) :-
   smallStep__2((B,D),conf((E,F))).
smallStep__2((mlpair_10_1(A,B),C,D),conf((mlpair_10_2(E,A),C))) :-
   smallStep__2((B,D),val(E)).
smallStep__2((mlpair_10_1(A,B),C,D),conf((mlpair_10_1(A,E),C,F))) :-
   smallStep__2((B,D),conf((E,F))).
smallStep__2((mlpair_10_2(A,B),C),val((A,D))) :-
   smallStep__2((B,C),val(D)).
smallStep__2((mlpair_10_2(A,B),C),conf((mlpair_10_2(A,D),E))) :-
   smallStep__2((B,C),conf((D,E))).
smallStep__2((apply_11_1(A,B),C,D),conf((apply_11_2(E,F,G,A),C))) :-
   smallStep__2((B,D),val(closure(E,F,G))).
smallStep__2((apply_11_1(A,B),C,D),conf((apply_11_1(A,E),C,F))) :-
   smallStep__2((B,D),conf((E,F))).
smallStep__2((apply_11_2(A,B,C,D),E),conf((B,[(A,F)|C]))) :-
   smallStep__2((D,E),val(F)).
smallStep__2((apply_11_2(A,B,C,D),E),conf((apply_11_2(A,B,C,F),G))) :-
   smallStep__2((D,E),conf((F,G))).
smallStep__2((apply_12_1(A,B),C,D),conf((apply_12_2(E,A),C))) :-
   smallStep__2((B,D),val(opaque(E))).
smallStep__2((apply_12_1(A,B),C,D),conf((apply_12_1(A,E),C,F))) :-
   smallStep__2((B,D),conf((E,F))).
smallStep__2((apply_12_2(A,B),C),val(D)) :-
   smallStep__2((B,C),val(E)),
   eval__5(A,E,D).
smallStep__2((apply_12_2(A,B),C),conf((apply_12_2(A,D),E))) :-
   smallStep__2((B,C),conf((D,E))).
smallStep__2((let_13_1(A,B,C),D,E),conf((B,[(A,F)|D]))) :-
   smallStep__2((C,E),val(F)).
smallStep__2((let_13_1(A,B,C),D,E),conf((let_13_1(A,B,F),D,G))) :-
   smallStep__2((C,E),conf((F,G))).
smallStep__2((letrec_14_1(A,B,C),D,E),conf((B,D))) :-
   smallStep__2((C,E),val(A)).
smallStep__2((letrec_14_1(A,B,C),D,E),conf((letrec_14_1(A,B,F),D,G))) :-
   smallStep__2((C,E),conf((F,G))).
\end{lstlisting}
\end{tabular}

\caption{Small-step rules for Mini-ML from the specialised interpreter.}\label{miniml-small-Horn}
\end{figure}

\end{example}

\subsection{Case study:  the reversible programming language Janus}\label{janus}

Big-step semantics for the reversible programming language Janus were written in 2007 
by Yokoyama and Gl\"{u}ck \cite{YokoyamaG07}.
Janus is ``an imperative, sequential language for reversibility" of programs.
In 2024, Lami \emph{et al.} \cite{LamiLS24} noted that no small-step semantics for Janus existed, and 
developed a set of small-step semantics rules, proving them equivalent to the earlier big-step semantics.

We rewrote the same big-step semantics as used by Lami \emph{et al.} as Horn clauses. They use 
two big-step transition relations, one for expressions and one for statements, and these
are merged into a single transition relation in our Horn clause representation.  
For statements, the environment
consists of two components, the store and the set of procedure definitions, whereas
for expressions, the second component is not needed. Specialising the same contextual small-step 
interpreter that we used for all our experiments, with respect to the Janus semantic rules,
we obtained essentially the same small-step rules as described by Lami \emph{et al.}
Furthermore, the proof of correctness follows from the correctness of the interpreter (shown in this paper)
and the specialiser.  For the latter we used \textsc{Logen} \cite{logen} in all our experiments, a robust
and well-tested partial evaluator for Prolog based on proven Horn clause transformations.

The full set of Horn clauses for the big-step rules, and the derived small-set rules, is shown in Appendix \ref{secA2}.
Here we discuss some of the key points, relating our result to the manually written rules developed by
Lami \emph{et al.}.

Consider the big-step rules for loops in Janus.  Figure \ref{janus-loop} shows the three rules needed,
whose names we abbreviate as \textsc{L1}, \textsc{L3} and \textsc{L3}.

Let us examine some of the small-step rules generated by specialising the interpreter.
In our Horn clauses, the construct $(e_1,s_1,e_2,s_2)$ appearing in the rules
is named $\mathtt{loop(E_1,S_1,E_2,S_2)}$.
The rules \textsc{L1} for the $\mathtt{from}$ constructor and \textsc{L3} for $\mathtt{loop}$
having multiple premises. give rise to a number of
suffix rules, whose syntax constructor is given names $\mathtt{from\_20\_1}$, $\mathtt{from\_20\_2}$, 
$\mathtt{loop\_21\_1}$, $\mathtt{loop\_22\_1}$, and so on.

Figure \ref{janus-small-Horn} shows the 20 small-step rules dealing with the
constructors $\mathtt{from}$ and $\mathtt{loop}$ and their suffix constructors.  These
correspond to two rules each for the 7 (non-final) suffixes of the rules,
plus 6 rules, two each for the original rules. Let us take one example to
illustrate their effect.  Each pairs of rules correspond to an instance of
the \emph{contextual rule} used by Lami \emph{et al.}  along with a rule for
moving from one evaluated component to the next.

 The following rules are for $\mathtt{loop\_22\_2}$,
which concern the second suffix of rule \textsc{L3}, namely $\rho\vdash s_2  \Downarrow  \rho'$.

\begin{lstlisting}[basicstyle=\footnotesize\ttfamily]
smallStep__2((loop_22_2(A,B,C,D,E),F,G),conf((loop_22_3(D,A,B,C,B),F,H,H))) :-
   smallStep__2((E,G),val(H)).
smallStep__2((loop_22_2(A,B,C,D,E),F,G),conf((loop_22_2(A,B,C,D,H),F,I))) :-
   smallStep__2((E,G),conf((H,I))).
\end{lstlisting}
The translation to standard notation is
straightforward, once the identities of the variables \texttt{A,B,C,...} in the 
specialised clauses have been established\footnote{
Unfortunately the partial evaluator \textsc{Logen} does not retain variable names from the source program.
We are working on a version that does so, which will make the small-step rules much easier to 
interpret.}.
\[
\dfrac{\langle  s'_2 ,\rho'\rangle \Rightarrow \rho'}
{ \langle \mathrm{loop\_22\_2}(e_2,e_1,s_1.s_2,s'_2),(\Gamma,\rho)\rangle \Rightarrow 
\langle \mathrm{loop\_22\_3}(s_2,e_2,e_1,s_1,s_1),(\Gamma,\rho',\rho') \rangle}\\
\]
\[
\dfrac{\langle s'_2,\rho\rangle \Rightarrow \langle s''_2,\rho''\rangle}
{\langle\mathrm{loop\_22\_2}(e_2,e_1,s_1,s_2,s'_2),(\Gamma,\rho)\rangle \Rightarrow 
\langle\mathrm{loop\_22\_2}(e_2,e_1,s_1,s_2,s''_2),(\Gamma,\rho'')\rangle}
\]
\vskip 0.3cm
The environment component $\Gamma$ stands for the set of procedure definitions, which is
carried through the small-step rules for statements.
The premise of the first rule, namely \texttt{smallStep\_\_2((E,G),val(H))}, represents the
termination of the execution of $\rho \vdash s_2  \Downarrow  \rho'$, and the
small step results in the next configuration \texttt{loop\_22\_3(D,A,B,C,B),F,H,H))},
which is the following suffix of the rule.
The premise of the second rule executes one small step of $\rho \vdash s_2  \Downarrow  \rho'$
and then continues with the same suffix, \texttt{(loop\_22\_2(A,B,C,D,H),F,I)}
These rules can be roughly related to the ones designed by Lami \emph{et al.}.  In their work 
they do not show all rules explicitly, but provide a schema based on \emph{evaluation contexts}
and \emph{contextual rules}.  The evaluation contexts for, say, the loop
construct contains a hole ($\bullet$) in each position of the
statement containing an expression or statement. 
The contextual rules successively perform small-step replacements in these
positions until that position is reduced to a value (the empty statement $\mathrm{skip}$
is regarded as a fully reduced statement) and then move to the next context ``hole".
Roughly speaking the contexts correspond to suffixes of the rule premises.
Rather than invent new syntax constructors as in our method, Lami \emph{et al.}
``overload" the existing syntax constructors with extra arguments, duplicating 
constructs that are needed in later premises.  Different overloadings arise as context holes are successively
reduced.

We did not attempt yet to produce the error rules developed by Lami \emph{et al.}.
In future work we intend to augment the Horn clause interpreter to handle failure of
a derivation explicitly.  We conjecture that this will allow us to generate small-step error rules
using the same specialisation techniques,

\begin{figure}
\[
\begin{array}{c}
\textsc{L1}~~~\dfrac{\rho \vdash e_1  \Downarrow  v_1~~~\mathrm{is\_true?}(v_1)~~~\rho \vdash s_1  \Downarrow  \rho'~~~\rho' \vdash (e_1,s_1,e_2,s_2) \Downarrow \rho''}
{\rho \vdash \mathrm{from} ~e_1~ \mathrm{do} ~s_1~ \mathrm{loop} ~s_2~ \mathrm{until} ~e_2 ~\Downarrow \rho''}\\
\\
\textsc{L2}~~~~~~~~~\dfrac{\rho \vdash e_2  \Downarrow  v_2~~~\mathrm{is\_true?}(v_2)}
{\rho \vdash (e_1,s_1,e_2,s_2) \Downarrow \rho}\\
\\
\rho \vdash e_2  \Downarrow  v_2~~~\mathrm{is\_false?}(v_2)~~~\rho \vdash s_2  \Downarrow  \rho'\\
\textsc{L3}~~~\dfrac{\rho' \vdash e_1  \Downarrow  v_1~~~\mathrm{is\_false?}(v_1)
~~~\rho' \vdash s_1  \Downarrow  \rho''
~~~\rho'' \vdash (e_1,s_1,e_2,s_2) \Downarrow \rho'''}
{\rho \vdash (e_1,s_1,e_2,s_2) \Downarrow \rho'''}\\
\\

 \end{array}
\]
\caption{Big-step rules for the loop construct in Janus}\label{janus-loop}
\end{figure}

\begin{figure}
\begin{tabular}{l}
\begin{lstlisting}[basicstyle=\footnotesize\ttfamily]
smallStep__2((from(A,B,C,D),E,F),conf((from_20_2(A,B,C,D,B),E,E,F))) :-
   smallStep__2((A,F),val(1)).
smallStep__2((from(A,B,C,D),E,F),conf((from_20_1(A,B,C,D,G),E,F,H))) :-
   smallStep__2((A,F),conf((G,H))).
smallStep__2((loop(A,B,C,D),E,F),val(F)) :-
   smallStep__2((C,F),val(1)).
smallStep__2((loop(A,B,C,D),E,F),conf((loop_21_1(G),F,H))) :-
   smallStep__2((C,F),conf((G,H))).
smallStep__2((loop(A,B,C,D),E,F),conf((loop_22_2(C,A,B,D,D),E,E,F))) :-
   smallStep__2((C,F),val(0)).
smallStep__2((loop(A,B,C,D),E,F),conf((loop_22_1(A,B,C,D,G),E,F,H))) :-
   smallStep__2((C,F),conf((G,H))).
smallStep__2((from_20_1(A,B,C,D,E),F,G,H),conf((from_20_2(A,B,C,D,B),F,F,G))) :-
   smallStep__2((E,H),val(1)).
smallStep__2((from_20_1(A,B,C,D,E),F,G,H),conf((from_20_1(A,B,C,D,I),F,G,J))) :-
   smallStep__2((E,H),conf((I,J))).
smallStep__2((from_20_2(A,B,C,D,E),F,G),conf((loop(A,B,D,C),F,H))) :-
   smallStep__2((E,G),val(H)).
smallStep__2((from_20_2(A,B,C,D,E),F,G),conf((from_20_2(A,B,C,D,H),F,I))) :-
   smallStep__2((E,G),conf((H,I))).
smallStep__2((loop_21_1(A),B,C),val(B)) :-
   smallStep__2((A,C),val(1)).
smallStep__2((loop_21_1(A),B,C),conf((loop_21_1(D),B,E))) :-
   smallStep__2((A,C),conf((D,E))).
smallStep__2((loop_22_1(A,B,C,D,E),F,G,H),conf((loop_22_2(C,A,B,D,D),F,F,G))) :-
   smallStep__2((E,H),val(0)).
smallStep__2((loop_22_1(A,B,C,D,E),F,G,H),conf((loop_22_1(A,B,C,D,I),F,G,J))) :-
   smallStep__2((E,H),conf((I,J))).
smallStep__2((loop_22_2(A,B,C,D,E),F,G),conf((loop_22_3(D,A,B,C,B),F,H,H))) :-
   smallStep__2((E,G),val(H)).
smallStep__2((loop_22_2(A,B,C,D,E),F,G),conf((loop_22_2(A,B,C,D,H),F,I))) :-
   smallStep__2((E,G),conf((H,I))).
smallStep__2((loop_22_3(A,B,C,D,E),F,G,H),conf((loop_22_4(C,A,B,D,D),F,F,G))) :-
   smallStep__2((E,H),val(0)).
smallStep__2((loop_22_3(A,B,C,D,E),F,G,H),conf((loop_22_3(A,B,C,D,I),F,G,J))) :-
   smallStep__2((E,H),conf((I,J))).
smallStep__2((loop_22_4(A,B,C,D,E),F,G),conf((loop(A,D,C,B),F,H))) :-
   smallStep__2((E,G),val(H)).
smallStep__2((loop_22_4(A,B,C,D,E),F,G),conf((loop_22_4(A,B,C,D,H),F,I))) :-
   smallStep__2((E,G),conf((H,I))).
\end{lstlisting}
\end{tabular}

\caption{Small-step rules for Janus loops.}\label{janus-small-Horn}
\end{figure}

\section{Discussion}\label{concl}

\paragraph{Big-step and small-step Horn clause derivations} We presented a formalisation of Horn clause derivations, in the style of operational
semantics.  Three different kinds of derivation were formalised: two small-step derivations
and one big-step.  
Two of them --  SLD small-step derivations and  big-step derivations -- were already
well known in various forms, for example as SLD proofs/refutations or AND-tree derivations.
The third kind, contextual small-step derivations, is a novel formulation of Horn clause
derivations, as far as we are aware.   All three forms were proved to be equivalent with
respect to successful derivations.

Another novelty regarding Horn clause derivations was the observation that the ``stack" or
conjunction of atoms awaiting selection during a small-step derivation has a particular form
that lends itself to ``folding" into a single atom, called $\mathtt{stack(z)}$, employing
``suffix rules" defining the predicate $\mathtt{stack(z)}$ formed from suffixes of
the given set of Horn clauses.

\paragraph{Application to big-step operational rules}
At this point these novelties might appear to be curiosities without practical value.
Their significance arises in derivations using a particular class of Horn clauses, namely
big-step semantic rules.  The structure of contextual small-step derivations turns
out to match small-step derivations in operational semantics. Refinements
to the generic Horn clause contextual small-step derivations, using features of big-step 
semantic rules, allows the connection
to be even closer.  In particular, the ``stack" predicate can be modelled by the big-step predicate,
and its value argument omitted from small steps.

\paragraph{Automatic translation of big-step to small-step semantics}
We then exploited the well-known method of interpreter specialisation to transform
big-step operational semantics to small-step semantics.  A contextual small-step derivation
interpreter is specialised with respect to a given set of big-step rules.  In effect
the contextual small-step derivation structure is ``compiled in" to the big-step rules,
resulting in clauses that are isomorphic to the desired small-step semantic rules.
This technique was demonstrated on a small imperative language, the functional language Mini-ML
with a $\mathrm{letrec}$ construct,
and the reversible language Janus.

\paragraph{Horn clauses for operational semantics}
Underlying the whole procedure is the use of Horn clauses to implement semantic rules.  Indeed, our
perspective is that operational semantics rules \emph{are} Horn clauses, however they are presented
on the page.  The declarative and procedural readings of Horn clauses are inherited by
semantic rules, rendering unnecessary the construction of ``definitional interpreters" to give
a procedural reading to semantic rules.  In our experiments, we were able to test the  equivalence
of the various semantics just by ``running" the respective Horn clauses on example programs
using a logic programming
system.  For instance, the Janus big-step rules and derived small-step rules could both be run on the
Fibonacci example of the Janus paper, with identical results of course;  similarly all of  the Mini-ML examples
presented by Kahn  \cite{Kahn87} could be tested in both semantics.  This is more than just convenience -- it underlines
the close logical connection between semantic rules and Horn clauses.
We have no doubt that there are many other ways of exploiting 
the very close connection between semantic rules and Horn clauses.(see Future Work).

\subsection{Related work}
There have been previous works demonstrating  automatic translation from big-step semantics 
to small-step semantics (or to abstract machines, a related problem)
\cite{VeselyF19,AmbalLSN22,HuizingKK10-short,Ager04-short}.
These approaches are based either on specifically designed transformations applied
to the big-step rules, or on transformations applied to a ``definitional interpreter" for the rules.
The latter approach is in the same spirit as ours, in that the transformed semantic rules
are ``inherited" from the transformed interpreter.  In particular, the ``continuation-passing"
transformation of a definitional interpreter \cite{VeselyF19,AmbalLSN22}, has clear connections to
our small steps with a stack or ``continuation". Some of the transformations 
applied in other work also have the flavour of partial evaluation, and the partial
evaluator that we used performs some of the transformations (such as argument filtering) seen
in other methods.
In summary, our approach, though not in principle more general, seems more
direct and transparent, in that the small-step derivation structure appears already in the
interpreter, which is then specialised with respect to a set of big-step rules.  
Implementation of the transformation using an existing program specialiser 
is also an advantage. The proof of correctness is
more direct;  the proofs of equivalence of the various Horn clause derivations
are straightforward inductions on derivation tree height.
We have succeeded in transforming some medium-sized examples, such as the Janus semantics,
that in their complexity go beyond examples we have seen in other work.
The correctness of our approach follows from the correctness of the generic Horn clause
interpreters and the correctness of the specialisation tools.

Our own previous work \cite{GallagherHML23} was also based on interpreter specialisation, but the interpreter
was an ad hoc procedure directed at big-step rules, unlike our current work, in which
the interpreters are for Horn clauses generally, and syntactic refinement are made for the
case of interpreting semantic rules.  In our previous work, correctness of the interpreter was not established; now
the correctness proofs rest on the correctness of the Horn clause derivation strategies.

The ``suffix rules" presented as part of our method suggest connections to other work on
manipulating big-step semantic rules.  The ``pretty-big-step" semantics of Chargu\'{e}raud 
\cite{Chargueraud13} proposes a decomposition of big-step rules into a
number of smaller rules dealing with one premise at a time.  
We postulate that pretty-big-step semantic rules could be
constructed by folding the big-step rules premises using our suffix rules;  but confirming this
guess remains future work.  A similar remark applies to the ``computation tree semantics"
recently proposed by Harbo and H\"{u}ttel \cite{HarboH26}.

\subsection{Future work}

Translation from small-step rules to big-step rules seems feasible using the same
approach.  A big-step Horn clause interpreter could be specialised with respect to 
a set of small step rules.  The main practical issue seems to be control of the
partial evaluation, which would be based on distinguishing between the ``run"
and ``small-step" predicates of the rules.

The Janus semantics discussed in Section \ref{specsem} prompts the question of
extending the translation to include rules
handling exceptions, abnormal termination and such like, as is done by Lami \emph{et al.}
Indeed, one of the reasons for using small-step semantics is the ability to observe more
than can be observed in a big-step derivation.  
We intend to derive such error rules from a Horn clause interpreter modified to
make failure observable.  Inspiration can be drawn from meta-interpreters for Prolog
capturing 
extensions of standard program execution, including failure states.

In program verification based on Constrained Horn Clauses (CHC) \cite{AngelisFGHPP22, BjornerGMR15}, it is required to
translate a program from some source language into CHCs.  This can be done reliably
using specialisation of semantics for the source language \cite{AngelisFGHPP22}.  The use of
big-step or small-step semantics can make an important difference to the CHCs obtained,
especially in the handling of loops in the source programs.  An automatic 
translator from big-step to small-step semantics and vice versa would
allow flexible choice of semantics; we see it as an important
component in a CHC-based verification toolchain.

\paragraph{Acknowledgements.} 
We gratefully acknowledge discussions with Robert Gl\"{u}ck (who suggested the Janus case study), Kristian Harbo, Hans H\"{u}ttel, Bishoksan Kafle, Morten Rhiger and Mads Rosendahl.

\bibliographystyle{splncs04}

\newpage
\begin{appendices}

\section{Proofs of Theorems}\label{secA1}

A run $R\stackrel{1}{\Rightarrow}\!\!*  ~R'$ is 
\emph{expanding} if no clause of the form
$B \leftarrow \true$ is used in the derivation. We write
$R\stackrel{1}{\Rightarrow}_e\!\!*  ~R'$ for an expanding run.

\noindent
\textbf{Lemma \ref{expanding-run}.}
For all expanding runs $R\stackrel{1}{\Rightarrow}_e\!\!*~R'$ where $R\neq \true$,
there exist $B,S,S'$ such that $R=(B,S)$ and $R'=(S',S)$.  
\begin{proof}
By induction on the height of run derivations.
\begin{itemize}
\item $k=1$:  
The run is $R\stackrel{1}{\Rightarrow}_e\!\!*  ~R$, where $R \neq \true$, hence
we write $R=(B,S)$ and  $(B,S) \stackrel{1}{\Rightarrow}_e\!\!*  ~(B,S)$, so the proposition
holds with $S'=B$.
\item $k\ge 1$:   
An expanding run of height 
$k+1$ exists only if there is a clause $B \leftarrow (B',S') \in [P]$
and a run of height $k$ $(B',S',S) \stackrel{1}{\Rightarrow}_e\!\!* ~R'$.
In this case the run of height $k+1$ is $(B,S)\stackrel{1}{\Rightarrow}_e\!\!* ~R'$.
By the inductive hypothesis, there exists $S''$ such that $R'=(S'',S',S)$, hence we have
$(B,S)\stackrel{1}{\Rightarrow}_e\!\!* ~(S'',S',S)$ which establishes the proposition.

\end{itemize}
\end{proof}

\noindent
\textbf{Lemma \ref{opsem1-lemma}. } 
For all conjunctions $R, R',$ and  $R''$, where $R\neq\true$,
$R\stackrel{1}{\Rightarrow}_e\!\!*  ~R'$  if and only if
$(R,R'') \stackrel{1}{\Rightarrow}_e \!\! * ~(R',R'')$, and both runs have the same height.

\begin{proof}
By induction on the height of run derivations using $\mathtt{[RUN0]}$ and $\mathtt{[RUN1]}$.
\begin{itemize}
\item $k=1$:  

Using $\mathtt{[RUN0]}$, $R\stackrel{1}{\Rightarrow}\!\!*  ~R$ and  $(R,R'')\stackrel{1}{\Rightarrow}\!\!*  ~(R,R'')$
both hold, for all conjunctions $R,R''$.  Both are expanding runs since no clause $B \leftarrow \true$
is used. Hence the property holds for $k=1$.
\item $k\ge 1$:  
Let $R\stackrel{1}{\Rightarrow}_e\!\!*  ~R'$ be a run of height $k$ 
(where this run is the second premise of $\mathtt{[RUN1]}$).  
An expanding run of height $k+1$ exists using this premise if and only if
there is a clause $B \leftarrow S$, $S\neq \true$ and $S$ is a prefix of $R$ (say $R=(S,S')$),
that is, there is a run $(B,S')\stackrel{1}{\Rightarrow}_e \!\!* ~R'$.
This holds if and only if there is 
an expanding run of height $k+1$, 
$(B,S',R'')\stackrel{1}{\Rightarrow}_e\!\!*  ~(R',R'')$ constructed using the same clause
and the premise  
$(R,R'') \stackrel{1}{\Rightarrow}_e\!\!* ~(R',R'')$ which is also a run of height $k$ (for any $R''$)
by the inductive hypothesis. 
\end{itemize}
\end{proof}

\noindent
\textbf{Lemma \ref{opsem1}. }
There is a contextual small-step transition $R\stackrel{2}{\Rightarrow} R'$ if and only if  there is a run $R  \stackrel{1}{\Rightarrow}_e\!\!*  ~(B,R' )\stackrel{1}{\Rightarrow} R'$, that is, an expanding run followed by a step using a clause $B\leftarrow \true$.
\begin{proof}
Separate proofs for the if and only if directions. 

\noindent
$\rightarrow$:
By induction on the height $k \ge 1$ of the derivation tree for $R  \stackrel{2}{\Rightarrow} R'$.
\begin{itemize}
\item $k=1$:  
\[
\begin{array}{llll}
&B  \stackrel{2}{\Rightarrow} \true~~
&\text{using }\mathtt{[RES0]} \text{ with } B \leftarrow \true \in [P]\\
\rightarrow~~& (B,\true) \stackrel{1}{\Rightarrow} (\true,\true)~~~&\text{using }\mathtt{[RES]}\\
\rightarrow~~& B \stackrel{1}{\Rightarrow}_e\!\! *~ B \stackrel{1}{\Rightarrow} \true~~~&
\text{using }\mathtt{[RUN1]},\mathtt{[RUN0]}.\\
\end{array}
\]

\item $k \ge 1$:
Assume that  the proposition holds for derivations of height  $k$.  
The $k+1^{th}$
step of a derivation of $R  \stackrel{2}{\Rightarrow} R'$ uses 
$\mathtt{[RES1}]$, $\mathtt{[CONJ0]}$ or $\mathtt{[CONJ1]}$.  

\[
\begin{array}{llll}

\mathtt{[RES1}]& &B \stackrel{2}{\Rightarrow} R_1 & B \leftarrow R  \in_k [P], R\neq\true\\
&&R \stackrel{1}{\Rightarrow}_e\!\!*  ~(B',R_1) \stackrel{1}{\Rightarrow} R_1&\text{by ind. hyp. and } B'\leftarrow \true \in [P]\\
&&B \stackrel{1}{\Rightarrow}  ~R&\text{using }\mathtt{[RES}]\\
& \rightarrow &B \stackrel{1}{\Rightarrow}_e\!\!*  ~(B',R_1)&\text{using } \mathtt{[RUN1}]\\
& \rightarrow &B \stackrel{1}{\Rightarrow}_e\!\!*  ~(B',R_1) \stackrel{1}{\Rightarrow} R_1&\\
\mathtt{[CONJ0]}& &R_1 \wedge R_2\stackrel{2}{\Rightarrow}  R_2&\\
&&R_1 \stackrel{1}{\Rightarrow}_e\!\!*  ~ B' \stackrel{1}{\Rightarrow}\true&\text{by ind. hyp. and }B'\leftarrow \true \in [P]\\
& \rightarrow &(R_1,R_2) \stackrel{1}{\Rightarrow}_e\!\!*  ~ (B' ,R_2)\stackrel{1}{\Rightarrow}R_2&\text{by Lemma \ref{opsem1-lemma}}\\
\mathtt{[CONJ1]}& &R_1 \wedge R_2 \stackrel{2}{\Rightarrow}  R'_1 \wedge R_2&\\
&&R_1 \stackrel{1}{\Rightarrow}_e\!\!*  ~(B',R'_1) \stackrel{1}{\Rightarrow} R'_1&\text{by ind. hyp. and }B'\leftarrow \true \in [P]\\
& \rightarrow &(R_1,R_2) \stackrel{1}{\Rightarrow}\!\!*  ~(B',R'_1,R_2)\stackrel{1}{\Rightarrow} (R'_1,R_2)&\text{by Lemma \ref{opsem1-lemma}}\\
\end{array}
\]

\end{itemize}

\noindent
$\leftarrow$: By induction on the height $k$ of the $\stackrel{1}{\Rightarrow}\!\!*$ run.
\begin{itemize}
\item
$k=1$: $R=(B,R')$ and the run has the form $(B,R') \stackrel{1}{\Rightarrow}_e\!\!*  ~(B,R') 
\stackrel{1}{\Rightarrow} R'$
using $\mathtt{[RUN0}]$ and $\mathtt{[RUN1}]$, where there is a clause
$B \leftarrow \true \in [P]$. 

If $R' = \true$ there is a derivation
of $(B \wedge R') \stackrel{2}{\Rightarrow} R'$ using $\mathtt{[RES0}]$.  If $R'\neq \true$ there is a derivation of
$(B \wedge R') \stackrel{2}{\Rightarrow} R'$ using $\mathtt{[CONJ0}]$, since $B \stackrel{2}{\Rightarrow} \true$ using 
$\mathtt{[RES0}]$.  

\item $k \ge 1$:  Assume that the property holds for runs of height $k$. 

\noindent
Let
$R\stackrel{1}{\Rightarrow}_e\!\!* ~(B'',R') \stackrel{1}{\Rightarrow} R'$ be a run of height $k$,
where $B''\leftarrow \true \in [P]$.
Using Lemma \ref{expanding-run}, this can be written as
$(B',S')\stackrel{1}{\Rightarrow}_e\!\!* ~(B'',S'',S') \stackrel{1}{\Rightarrow} (S'',S') $, where 
$R=(B',S')$ and $R'=(S'',S')$.
Let there be a clause $B \leftarrow B',R''\in [P]$, where $S'=(R'',S''')$.
Then 
$(B,S''') \stackrel{1}{\Rightarrow}_e\!\!* ~(B'',S'',R'',S''') \stackrel{1}{\Rightarrow} (S'',R'',S''')$ 
is a run of height $k+1$,
constructed using 
$\mathtt{[RUN1}]$. We now show that 
$(B \wedge S''')  \stackrel{2}{\Rightarrow}(S'' \wedge R'' )\wedge S''' $.
\[
\begin{array}{lll}
\\
&(B',S')\stackrel{1}{\Rightarrow}_e\!\!* ~(B'',S'',S') &\text{is an expanding run}\\
&B'\stackrel{1}{\Rightarrow}_e\!\!* ~(B'',S'') &\text{is an expanding run, by Lemma \ref{opsem1-lemma} }\\
&B'\stackrel{2}{\Rightarrow} S''&B'' \leftarrow \true\text{ and using  ind. hyp. }\\
\rightarrow&(B' \wedge R'')\stackrel{2}{\Rightarrow} (S'' \wedge R'')&\text{using rule }\mathtt{[CONJ1}]\\
\rightarrow&B\stackrel{2}{\Rightarrow} (S'' \wedge R")&\text{using rule }\mathtt{[RES1}]\\
&&\text{and clause }B \leftarrow B',R''\in [P]\\
\rightarrow&(B \wedge S''')  \stackrel{2}{\Rightarrow}(S'' \wedge R'' ) \wedge S''' &\text{using rule }\mathtt{[CONJ1}]\\
\\

\end{array}
\]

\end{itemize}

\end{proof}
\noindent
\textbf{Theorem \ref{opsem2}.}

For a given set of Horn clauses, $R \stackrel{1}{\Rightarrow}\!\!* ~\true$ 
if and only if 
$R \stackrel{2}{\Rightarrow}\!\!*  ~\true $.
\begin{proof} Separate cases for the \emph{if} and \emph{only if }implications.

\begin{itemize}

\item 
$R \stackrel{1}{\Rightarrow}\!\!* ~\true$ implies $R \stackrel{2}{\Rightarrow}\!\!*  ~\true$.
Let $k$ be the number of steps in the run $R \stackrel{1}{\Rightarrow}\!\!* ~\true$ that 
use a clause of form 
$B \leftarrow \true$. Clearly the final step in the run uses such a clause.
Thus the run consists of $k$ (possibly empty) expanding runs each followed by the use of
a clause of form 
$B \leftarrow \true$.
By Lemma \ref{opsem1}, each such section corresponds
to a $\stackrel{2}{\Rightarrow}$ step. By induction on $k$ we can construct a run 
$R \stackrel{2}{\Rightarrow}\!\!*  ~\true$ of length $k$.

\item
$R \stackrel{2}{\Rightarrow}\!\!*  ~\true$ implies $R \stackrel{1}{\Rightarrow}\!\!* ~\true$. 
Each
step $R \stackrel{2}{\Rightarrow}  R'$  corresponds to a run $R \stackrel{1}{\Rightarrow}\!\!*  ~R'$
using Lemma \ref{opsem1}.  By induction on the height of the $\stackrel{2}{\Rightarrow}\!\!*$ run
we can concatenate the runs into a single run $R \stackrel{1}{\Rightarrow}\!\!* ~\true$. 
\end{itemize}

\end{proof}

\noindent
\textbf{Lemma \ref{folding2}.}
Let $S_1, S_2$ be foldable conjunctions or $\true$.  Then
\[
S_1 \cong S_2 \text{ if and only if }\fold(S_1) \cong \fold(S_2). 
\]
\begin{proof}
If $S_1=S_2=\true$ the result is trivial.  Otherwise, 
let $S_1\cong S_2$ be (a renaming of)
\[
\C^n=((\cdots(R^{k_1}_1 \wedge R_2)^{k_2} \wedge \cdots R_{n-1})^{k_{n-1}}\wedge R_n)^{k_n}.
\]
The proof is by induction on $n\ge 1$.
\begin{itemize}
\item $n=1$:  $\C^1=R^{k_1}_1$. 
$\fold(R^{k_1}_1)= \mathtt{stack(t)}$, where $\mathtt{stack(t)}$ is an instance of the 
head of suffix clause $k_1$
uniquely determined by $R_1$, $k_1$, the length of $R_1$ and whether or not the first atom in
$R_1$ has predicate $\mathtt{stack}$ (determining type 1 or 2).
Furthermore, $\mathtt{stack(t)}$ uniquely determines
the instance of its body, $R_1$ (up to renaming of local variables in $R_1$, which are not in $\mathtt{t}$),
since the function symbol in $\mathtt{t}$ is unique to suffix clause $k_1$ (of type 1 or 2).
\item $n\ge 1$: 
$S_1=S_2=(\C^{n-1} \wedge R_n)^{k_n}$.  By the inductive hypothesis. $\fold(\C^{n-1})$
uniquely determines its result $\mathtt{stack(t)}$.  Then evaluating $\fold((\mathtt{stack(t)} \wedge R_n)^{k_n})$
gives a unique result $\mathtt{stack(t')}$, determined by $k_n$ and the length of $R_n$
and it is of type 2 since the first atom in the suffix body is $\mathtt{stack(t)}$.
Furthermore, $\mathtt{stack(t')}$ uniquely determines the instance $(\mathtt{stack(t)} \wedge R_n)$
(up to renaming of local variables which are not in $\mathtt{t'}$), since the function symbol in 
$\mathtt{t'}$ is unique to suffix clause $k_n$ of type 2. 

\end{itemize}

\end{proof}

For the proof of Lemma \ref{foldable} and Theorem \ref{fig3-fig4} we introduce
the concept $\aof(R)$, meaning that $R$ is either an atom in the signature
of a given set of Horn clauses $P$ (a $P$-atom) or else foldable using the clauses $P_{\mathrm{suff}}$.
This notation allows more concise proofs.

The following property is used in the proofs.  Given a set of Horn clauses $P$:
\[
\begin{array}{ll}
\aof(R) \text{ implies } &R \text{ is a $P$-atom}\\
&\text{or  } \\
&R=R_1 \wedge R_2, R \text{ is foldable,}\\
&\aof(R_1) \text{ and }\\
&R_2 \text{ is a suffix of a clause in }P
\end{array}
\]

\noindent
\textbf{Lemma \ref{foldable}.}
Let $P$ be a set of Horn clauses, $B$ be an atom in the signature of $P$ and $R$ be 
foldable.  Then
 \begin{enumerate}
 \item
If $B\stackrel{2}{\Rightarrow} R_1$ then $R_1$ is a foldable conjunction or $\true$.
 \item
If $R\stackrel{2}{\Rightarrow} R_1$ then $R_1$ is a foldable conjunction or $\true$.
 \end{enumerate}
With the notation just introduced, we can combine the two cases of the Lemma.
Let $\aof(R)$. Then
\begin{itemize}
 \item
If $R\stackrel{2}{\Rightarrow} R_1$ then $R_1$ is a foldable conjunction or $\true$.
 \end{itemize}

\begin{proof}
Let $P_{\mathtt{suff}}$ be the set of suffix clauses for $P$. The proof is by
induction on the height $m\ge 1$ of the derivation.
\begin{itemize}
\item 
$m=1$: The derivation is $B\stackrel{2}{\Rightarrow} \true$ using rule $\mathtt{[RES0]}$,
so the property is established with $R_1=\true$.
\item
$m\ge 1$: Assume the property holds for derivations of height $m$.  Consider a derivation of
height $m+1$, whose final step is $B\stackrel{2}{\Rightarrow} \,R_1$ or
$R\stackrel{2}{\Rightarrow} R_1$ where $\aof(R)$.
The final
step of the derivation uses $\mathtt{[RES1}_1]$, $\mathtt{[RES1}_2]$, $\mathtt{[RES1}_3]$
or $\mathtt{[CONJ0]}$ or $\mathtt{[CONJ1]}$.
The following cases establish that $R_1$ is foldable or $\true$..
\[
\begin{array}{llll}
\mathtt{[RES1}_1]&& B\stackrel{2}{\Rightarrow} \,R_1&B \leftarrow B_1  \in [P]\\
& \rightarrow &B_1\stackrel{2}{\Rightarrow} \,R_1&B_1 \text{ is an atom}\\
& \rightarrow &R_1 \text{ is foldable or }\true&\text{by ind. hyp.}\\
\mathtt{[RES1}_2]& &B\stackrel{2}{\Rightarrow} \,R_1& B \leftarrow B_1 \wedge R_1  \in [P]\\
& \rightarrow &R_1  \text{ is foldable }&\text{since }R_1\text{  is a clause suffix}\\
\mathtt{[RES1}_3]& &B\stackrel{2}{\Rightarrow} \,R_1& R_1=(R_2 \wedge R), B \leftarrow B_1 \wedge R  \in [P]\\
&\rightarrow&B_1\stackrel{2}{\Rightarrow} R_2 &B_1 \text{ is an atom, }R_2 \neq \true\\
& \rightarrow &R_2  \text{ is foldable }&\text{by ind. hyp.}\\
& \rightarrow &(R_2 \wedge R)  \text{ is foldable }&\text{since }R \text{ is a clause suffix}\\
\mathtt{[CONJ0]}& &R\stackrel{2}{\Rightarrow} R_1&\text{where }R=(R_2 \wedge R_1)\\
& \rightarrow &R_1 \text{ is foldable }&R_1 \text{ is a suffix, since }\aof(R)\\
\mathtt{[CONJ1]}& &R\stackrel{2}{\Rightarrow} R_1&\text{where }R=(R_2 \wedge R'), R_1=(R_3 \wedge R')\\
&\rightarrow&\aof(R_2)&\\
&&\text{and }R' \text{ is a suffix}& \text{since }\aof(R)\\
&\rightarrow&R_2 \stackrel{2}{\Rightarrow} R_3&\text{premise of }\mathtt{[CONJ1]}, R_3 \neq \true\\
&\rightarrow&R_3 \text{ is foldable}&\text{by ind. hyp}\\
&\rightarrow&(R_3 \wedge R') \text{ is foldable}&\\
\end{array}
\]

\end{itemize}
Hence, the property holds for derivations of height $m$, for all $m\ge 1$.
\end{proof}

\noindent
\textbf{Theorem \ref{fig3-fig4}.
}Let $P$ be a set of Horn clauses, $B$ be an atom in the signature of $P$ and $R,R_1$ be 
foldable conjunctions.  Then
 \begin{enumerate}
 \item
$B\stackrel{2}{\Rightarrow} R_1$ if and only if $\exists k_1.B\stackrel{3}{\Rightarrow} \fold(R_1^{k_1})$.
 \item
 $R\stackrel{2}{\Rightarrow} R_1$ if and only if $\exists k_1,k.\fold(R^k)\stackrel{3}{\Rightarrow} \fold(R_1^{k_1})$.
 \end{enumerate}

\begin{proof}
We first prove the forward direction of the equivalences.
The proof is by
induction on the height $m\ge 1$ of the $\stackrel{2}{\Rightarrow}$ derivations.
\begin{itemize}
\item 
$m=1$: $B \stackrel{2}{\Rightarrow} \,\true$ using $\mathtt{[RES0]}$, where
there is a clause $B \leftarrow \true \in [P]$.  This holds if and only if there is
a derivation $B \stackrel{3}{\Rightarrow} \,\true$, that is $B \stackrel{3}{\Rightarrow} \,\fold(\true)$,
using $\mathtt{[UNF0]}$.
\item
$m\ge 1$: Assume the property holds for derivations of height $m$.  Consider a derivation of
height $m+1$, whose final step is $B \stackrel{2}{\Rightarrow} \,R_1$ or
$R \stackrel{2}{\Rightarrow} \,R_1$.  
The final
step of the derivation uses $\mathtt{[RES1}_1]$, $\mathtt{[RES1}_2]$, $\mathtt{[RES1}_3]$
(for $B \stackrel{2}{\Rightarrow} R_1$), or $\mathtt{[CONJ0]}$ or $\mathtt{[CONJ1]}$
(for $R \stackrel{2}{\Rightarrow} R_1$).
The following cases establish if this holds then
$\exists k_1 .B \stackrel{3}{\Rightarrow} \,\fold(R^{k_1}_1)$ or
$\exists k,k_1 .\fold(R^k) \stackrel{3}{\Rightarrow} \,\fold(R^{k_1}_1)$ respectively.
\[
\begin{array}{llll}
\mathtt{[RES1}_1]&& B\stackrel{2}\Rightarrow R&\\
& \rightarrow &B'\stackrel{2}\Rightarrow R&\text{where }B \leftarrow B'  \in [P]\\
& \rightarrow &\exists k.B'\stackrel{3}\Rightarrow \fold(R^k)&\text{by ind. hyp. }\\
& \rightarrow &\exists k. B\stackrel{3}\Rightarrow \fold(R^k)&\text{using }\mathtt{[UNF1]} \\
&&&  \text{and }B \leftarrow B'  \in [P \cup P_{\mathrm{suff}}]\\
\mathtt{[RES1}_2]&& B\stackrel{2}\Rightarrow R&\\
& \rightarrow &B'\stackrel{2}\Rightarrow \true&\text{where }B \leftarrow B' \wedge R  \in_k  [P]\\
& \rightarrow &B'\stackrel{3}\Rightarrow \true&\text{by ind. hyp., and }\fold(\true)=\true\\
& \rightarrow &\exists k.B\stackrel{3}\Rightarrow \fold(R^k)&\text{using }\mathtt{[FOLD0]}\\
 &&&\text{and }B \leftarrow B' \wedge R  \in_k   [P \cup P_{\mathrm{suff}}]\\
\mathtt{[RES1}_3]&& B\stackrel{2}\Rightarrow (R_1 \wedge R)&\\
& \rightarrow &B'\stackrel{2}\Rightarrow R_1&\text{where }B \leftarrow B' \wedge R  \in_k  [P],\\
&&& R_1\neq \true\\
& \rightarrow &\exists k_1.B'\stackrel{3}\Rightarrow \fold(R_1^{k_1})&\text{by ind. hyp. }\\
& \rightarrow &\exists k,k_1.B\stackrel{3}\Rightarrow \fold(\fold(R_1^{k_1}) \wedge R)^k)&
   \text{using }\mathtt{[FOLD1]} \\
 &&&\text{and }B \leftarrow B' \wedge R  \in_k  [P \cup P_{\mathrm{suff}}]\\
 & \rightarrow &\exists k,k_1.B\stackrel{3}\Rightarrow \fold(R_1^{k_1}\wedge R^k)&\\
 \end{array}
\]
\[
\begin{array}{llll}
\mathtt{[CONJ0]}&& (R_1 \wedge R)\stackrel{2}\Rightarrow R&(R_1 \wedge R) \text{ foldable, }\aof(R_1)\\
& \rightarrow &R_1 \stackrel{2}\Rightarrow\true&\text{premise of }\mathtt{[CONJ0]} \\
& \rightarrow &B'\stackrel{3}\Rightarrow \fold(\true)&\text{by ind. hyp. where } B'=\fold(R^{k_1}_1)\\
& &&\text{or } B' =R_1\text{ is an atom}\\
&&\exists k.B \leftarrow B'\wedge R \in_k [P_{\mathrm{suff}}]&
  \text{where }B= \fold(\fold(R_1^{k_1}) \wedge R)^k)\\
&&&\text{or }B= \fold(B'\wedge R)^k)\\

& \rightarrow &B\stackrel{3}\Rightarrow \fold(R^k)&\text{using }\mathtt{[FOLD0]}\\
& \rightarrow &\exists k,k_1.\fold(R_1^{k_1} \wedge R^k)\stackrel{3}\Rightarrow \fold(R^k)&\text{if }B= \fold(\fold(R_1^{k_1}) \wedge R)^k)\\
& &\text{or }\exists k.\fold(R_1 \wedge R)^k)\stackrel{3}\Rightarrow \fold(R^k)&\text{if }B= \fold(B'\wedge R)^k)\\
\mathtt{[CONJ1]}&& (R_1 \wedge R)\stackrel{2}\Rightarrow (R_2 \wedge R)&(R_1 \wedge R) \text{ foldable, }\aof(R_1)\\
& \rightarrow &R_1\stackrel{2}\Rightarrow R_2&\text{premise of }\mathtt{[CONJ1]} \\
& \rightarrow &\exists k_2.B'\stackrel{3}\Rightarrow \fold(R^{k_2}_2)&\text{by ind. hyp., where }
B'=\fold(R^{k_1}_1)\\
& &&\text{or }B' \text{ is an atom}\\
&&\exists k.B \leftarrow B'\wedge R \in_k [P_{\mathrm{suff}}]&
\text{where }B= \fold(\fold(R_1^{k_1}) \wedge R)^k)\\
&&&\text{or }B= \fold(B'\wedge R)^k)\\
& \rightarrow &\exists k_2.B \stackrel{3}\Rightarrow \fold(\fold(R_2^{k_2}) \wedge R)^k)&
\text{using }\mathtt{[FOLD1]} \\
& \rightarrow &\exists k,k_1,k_2. \fold(R_1^{k_1} \wedge R^k)&\\
&&~~~~~~~~~~~~ \stackrel{3}\Rightarrow \fold(R_2^{k_2} \wedge R^k)&\\
\end{array}
\]

\end{itemize}

Next we prove the reverse direction, by induction on the height $m \ge 1$ of the 
$\stackrel{3}\Rightarrow$ derivation.

\begin{itemize}
\item 
$m=1$: $B \stackrel{3}{\Rightarrow} \,\fold(\true)$ (where $\fold(\true)=\true$) using 
$\mathtt{[UNF0]}$, where
there is a clause $B \leftarrow \true \in [P]$.  This implies there is
a derivation $B \stackrel{2}{\Rightarrow} \true$
using $\mathtt{[RES0]}$.
\item
$m\ge 1$: Assume the property holds for derivations of height $m$.  Consider a derivation of
height $m+1$, whose final step is $B \stackrel{3}{\Rightarrow} \fold(R^{k_1}_{k_1})$, or
$\fold(R^k) \stackrel{3}{\Rightarrow} \fold(R^{k_1}_{k_1})$
and the final step uses $\mathtt{[UNF1]}$, $\mathtt{[FOLD0]}$ or $\mathtt{[FOLD1]}$.
The following cases establish that if this holds then 
$B \stackrel{2}{\Rightarrow} R_1$ or $R \stackrel{2}{\Rightarrow} R_1$
respectively.
First, consider the case $B \stackrel{3}{\Rightarrow} \fold(R^{k_1}_1)$.
\[
\begin{array}{llll}
\mathtt{[UNF1]}&& B\stackrel{3}\Rightarrow \fold(R_1^{k_1})&
\text{where }B \leftarrow B'  \in [P \cup P_{\mathrm{suff}}]\\
& \rightarrow &B'\stackrel{3}\Rightarrow \fold(R_1^{k_1})&\text{premise of }\mathtt{[UNF1]}\\
& \rightarrow &B'\stackrel{2}\Rightarrow R_1&\text{by ind. hyp. }\\
& \rightarrow &B\stackrel{2}\Rightarrow R_1&\text{using }\mathtt{[RES1}_1] \\
&&&  \text{and }B \leftarrow B'  \in [P]\\
\mathtt{[FOLD0]}&& B\stackrel{3}\Rightarrow \fold(R_1^{k_1})&
\text{where }B \leftarrow B' \wedge R_1  \in_k  [P \cup P_{\mathrm{suff}}]\\
& \rightarrow &B'\stackrel{3}\Rightarrow \true&\text{premise of }\mathtt{[FOLD0]}\\
	
& \rightarrow &B'\stackrel{2}\Rightarrow \true&\text{by ind. hyp. }\\
& \rightarrow &B\stackrel{2}\Rightarrow R_1&\text{using }\mathtt{[RES1}_2]\\
 &&&\text{and }B \leftarrow B' \wedge R_1  \in_k  \in [P]\\
\mathtt{[FOLD1]}&&B\stackrel{3}\Rightarrow \fold(\fold(R_1^{k_1}) \wedge R)^{k})&
\text{where }B \leftarrow B' \wedge R  \in_k  [P \cup P_{\mathrm{suff}}]\\
& \rightarrow &B'\stackrel{3}\Rightarrow \fold(R^{k_1}_1)&\\
& \rightarrow &B'\stackrel{2}\Rightarrow R_1&\text{by ind. hyp. }\\
& \rightarrow &B\stackrel{2}\Rightarrow R_1 \wedge R&
   \text{using }\mathtt{[RES1}_3] \\
 &&&\text{and }B \leftarrow B' \wedge R  \in_k  \in [P]\\
 \end{array}
\]
Now consider a derivation of height $m+1$ where the final step is 
$\fold(R^k)\stackrel{3}{\Rightarrow} \fold(R_1^{k_1})$.  The final step
can use 
$\mathtt{[FOLD0]}$ or $\mathtt{[FOLD1]}$.
The following property relating folding and suffix clauses is used.
\begin{itemize}
\item
$B \leftarrow R \in_k P_{\mathrm{suff}}$ is a suffix clause of type 1 if and only if $B=\fold(R^k)$.
\item
$B \leftarrow B' \wedge R \in_k P_{\mathrm{suff}}$ is a suffix clause of type 2 if and only if 
$B' = \fold(R^{k_1}_1)$ for some $R_1^{k_1}$ and $B=\fold(R^{k_1}_1 \wedge R^k)$.
\end{itemize}
\[
\begin{array}{llll}
\mathtt{[FOLD0]}&& \fold(R^k)\stackrel{3}\Rightarrow \fold(R_1^{k_1})&\text{let } \fold(R^k) =B,~
B \leftarrow B' \wedge R_1  \in_k  [P \cup P_{\mathrm{suff}}]\\
&&&\text{hence }R=R_2 \wedge R_1\\
&&&\text{where } \aof(R_2)\\
& \rightarrow &B'\stackrel{3}\Rightarrow \true&\text{premise of }\mathtt{[FOLD0]}\\
& \rightarrow &R_2\stackrel{2}\Rightarrow \true&\text{by ind. hyp. }\\
&  &&\text{either }B'=\fold(R_2) \text { or }R_2=B'\\

& \rightarrow &R_2 \wedge R_1\stackrel{2}\Rightarrow R_1&\text{using }\mathtt{[CONJ0]}\\
& \rightarrow & R \stackrel{2}\Rightarrow R_1 & \text{since }R=R_2 \wedge R_1\\
\mathtt{[FOLD1]}&& \fold(R^k)\stackrel{3}\Rightarrow \fold(R_1^{k})&\text{let } \fold(R^k) =B,~
B \leftarrow B' \wedge R_3  \in_k  [P \cup P_{\mathrm{suff}}]\\
&&&\text{hence }R=R_4 \wedge R_3\\
&&&\text{where } \aof(R_4)\\
&&&\text{and let } \fold(R_1^k)=\fold(\fold(R_2^{k_2}) \wedge R_3)^k)\\

& \rightarrow & B_1 \stackrel{3}\Rightarrow \fold(R^{k_2}_2 )&\text{premise of }\mathtt{[FOLD1]}\\
& \rightarrow &R_4\stackrel{2}\Rightarrow R_2&\text{by ind. hyp. }\\
&  &&\text{either }B'=\fold(R_4) \text { or }R_4=B'\\
& \rightarrow &R_4 \wedge R_3\stackrel{2}\Rightarrow R_2 \wedge R_3&\text{using }\mathtt{[CONJ1]}\\
& \rightarrow & R \stackrel{2}\Rightarrow R_1 & \text{since }R=R_4 \wedge R_3, R_1=R_2 \wedge R_3\\
 \end{array}
\]

\end{itemize}
\end{proof}

\noindent
\textbf{Lemma \ref{final-value-lemma}.}
Let $P$ be a set of big-step rules satisfying the final value condition (Definition \ref{finalvalue}).
Let $(\langle s_1,\rho_1\rangle \Downarrow v_1) \stackrel{3}{\Rightarrow} 
(\langle s_2,\rho_2\rangle \Downarrow v_2)$ be derived 
using the rules in Figure \ref{ss_v3}.
Then $v_1 = v_2$.

\begin{proof}
The shortest proof of $(\langle s_1,\rho_1\rangle \Downarrow v_1) \stackrel{3}{\Rightarrow} 
(\langle s_2,\rho_2\rangle \Downarrow v_2)$ using the rules
in Figure \ref{ss_v3} has height 2, using rules $\mathtt{[FOLD0]}$ and $\mathtt{[UNF0]}$.  The proof is by induction on 
derivations of height $k \ge 2$.
\begin{itemize}
\item $k=2$: 
Using $\mathtt{[FOLD0]}$, there are clauses;
\[
\begin{array}{l}
(\langle s,\rho\rangle \Downarrow v_1) \leftarrow (\langle s_2,\rho_2\rangle \Downarrow v_2),R_2.\\
(\langle s',\rho'\rangle \Downarrow v_2) \leftarrow R_2.\\
\end{array}
\]  
If the premise of $\mathtt{[FOLD0]}$, $(\langle s_2,\rho_2\rangle \Downarrow v_2)\stackrel{3}{\Rightarrow} \true$ 
succeeds using $\mathtt{[UNF0]}$,  the conclusion of $\mathtt{[UNF1]}$,
$(\langle s,\rho\rangle \Downarrow v_1)  \stackrel{3}{\Rightarrow} (\langle s',\rho'\rangle \Downarrow v_2)$ 
is derived.
Since both clause satisfy the final value condition and have the same suffix $R_2$, $v_1=v_2$. 

\item $k>2$:
Assume the property holds for derivations of height $k$.  In a derivation of height $k+1$, the final step uses 
$\mathtt{[UNF1]}$, $\mathtt{[FOLD0]}$ or $\mathtt{[FOLD1]}$.
\[
\begin{array}{llll}
\mathtt{[UNF1]}&&\text{Premise }(\langle s_1,\rho_1\rangle \Downarrow v_1)  \stackrel{3}{\Rightarrow} (\langle s_2,\rho_2\rangle \Downarrow v_2) ~
&B_1 \stackrel{3}{\Rightarrow}  B_2\\\
&&\text{Conclusion }(\langle s,\rho\rangle \Downarrow v)  
\stackrel{3}{\Rightarrow} (\langle s_2,\rho_2\rangle \Downarrow v_2)
&B \stackrel{3}{\Rightarrow}  B_2\\
&& \text{Clause }(\langle s,\rho\rangle \Downarrow v) \leftarrow (\langle s_2,\rho_2\rangle \Downarrow v) &B \leftarrow B_1\\
&&\text{ satisfies final value condition }&\\
&\rightarrow&v_1=v_2, \text{ by ind. hyp. and }v= v_1 \text{ hence }v_1=v_2&\\
\mathtt{[FOLD0]}&&\text{Conclusion }(\langle s,\rho\rangle \Downarrow v)  
\stackrel{3}{\Rightarrow} (\langle s',\rho'\rangle \Downarrow v')
&B \stackrel{3}{\Rightarrow}  B'\\
&& \text{Clauses }&\\
&&(\langle s,\rho\rangle \Downarrow v) \leftarrow (\langle s_1,\rho_1\rangle \Downarrow v_1),R_2&
B \leftarrow B_1,R_2\\
&&(\langle s',\rho'\rangle \Downarrow v') \leftarrow R_2&
B' \leftarrow R_2\\
&&\text{satisfy final value condition, have common suffix }R_2&\\
&\rightarrow&v=v'&\\
\mathtt{[FOLD1]}&&\text{Conclusion }(\langle s,\rho\rangle \Downarrow v)  
\stackrel{3}{\Rightarrow} (\langle s',\rho'\rangle \Downarrow v')
&B \stackrel{3}{\Rightarrow}  B'\\
&& \text{Clauses }&\\
&&(\langle s,\rho\rangle \Downarrow v) \leftarrow (\langle s_1,\rho_1\rangle \Downarrow v_1),R_2&
B \leftarrow B_1,R_2\\
&&(\langle s',\rho'\rangle \Downarrow v') \leftarrow (\langle s_3,\rho_3\rangle \Downarrow v_3),R_2&
B' \leftarrow B_3,R_2\\
&&\text{satisfy final value condition, have common suffix }R_2&\\
&\rightarrow&v=v'&\\
\end{array}
\]

\end{itemize}
\end{proof}

\section{Semantics of Janus}\label{secA2}

\subsection{Janus big-step semantics}

The big-step semantics for Janus is taken from Lam \emph{et al.} \cite{LamiLS24}.
The rules from Figs. 2, 3 and 4 of that paper are rewritten directly as Horn
clauses, shown below.  The code below can be run as a Prolog program and
used to interpret Janus programs.
\begin{lstlisting}[basicstyle=\footnotesize\ttfamily]
% Janus big-step semantics

:- module(janus,_).

% Lami, Lanese and Stefani. A Small-step semantics for Janus, 2024

%%%%%%%%%%%%%%%%%%%%%%
% Expressions. Fig. 2
%%%%%%%%%%%%%%%%%%%%%%


eval(var(X),St,V) :-
	find(St,X,V).
eval(const(N),_St,N).
%eval(array(X,E),St0,VV) :-
%	eval(E,St0,V),
%	arrayFind(St0,X,V,XV).
eval(add(E1,E2),St0,V) :-
	eval(E1,St0,V1),
	eval(E2,St0,V2),
	V is V1+V2.
eval(sub(E1,E2),St0,V) :-
	eval(E1,St0,V1),
	eval(E2,St0,V2),
	V is V1-V2.
eval(mul(E1,E2),St0,V) :-
	eval(E1,St0,V1),
	eval(E2,St0,V2),
	V is V1*V2.
eval(div(E1,E2),St0,V) :-
	eval(E1,St0,V1),
	eval(E2,St0,V2),
	V is V1/V2.
eval(E1>E2,St0,V) :-
	eval(E1,St0,V1),
	eval(E2,St0,V2),
	gt(V1,V2,V).
eval(E1<E2,St0,V) :-
	eval(E1,St0,V1),
	eval(E2,St0,V2),
	lt(V1,V2,V).
eval(E1>=E2,St0,V) :-
	eval(E1,St0,V1),
	eval(E2,St0,V2),
	gte(V1,V2,V).
eval(E1=<E2,St0,V) :-
	eval(E1,St0,V1),
	eval(E2,St0,V2),
	lte(V1,V2,V).
eval(E1==E2,St0,V) :-
	eval(E1,St0,V1),
	eval(E2,St0,V2),
	eq(V1,V2,V).
eval(logicaland(E1,E2),St0,V) :-
	eval(E1,St0,V1),
	eval(E2,St0,V2),
	V is V1/\V2.
eval(not(E),St0,V) :-
	eval(E,St0,V1),
	negate(V1,V).
	

	
%%%%%%%%%%%%%%%%%%%%%%%%%%%%%%	
%----- statements ----- Fig. 3
%%%%%%%%%%%%%%%%%%%%%%%%%%%%%%

eval(asg(var(X),Op,E),(_Env,St0),St1) :-			% ASVAR
	find(St0,X,XV),
	eval(E,St0,EV),
	evalOp(Op,XV,EV,St0,V),
	save(X,V,St0,St1).
%eval(asgarray(array(X,E1),Op,E),(_Env,St0),St1) :-	% ASSARRAY not used
%	eval(E1,St0,V1),
%	eval(E,St0,V),
%	arrayfind(St0,X,V1,XV1),
%	evalOp(Op,XV1,V,St0,V2),
%	arraysave(X,V1,V2,St0,St1).
eval(call(F),(Env,St),St1) :-						% CALL
	def(F,Env,S),
	eval(S,(Env,St),St1).
eval(uncall(F),(Env,St),St1) :-						% UNCALL
	def(F,Env,S),
	invert(S,S1),
	eval(S1,(Env,St),St1).
eval(seq(S1,S2),(Env,St),St2) :-					% SEQ
	eval(S1,(Env,St),St1),
	eval(S2,(Env,St1),St2).
eval(ifthenelse(E1,S1,_S2,E2),(Env,St),St2) :-		% IFTRUE
	eval(E1,St,1),
	eval(S1,(Env,St),St1),
	eval(E2,St,1),
	St2=St1.		% final value condition
eval(ifthenelse(E1,_S1,S2,E2),(Env,St),St2) :-		% IFFALSE
	eval(E1,St,0),
	eval(S2,(Env,St),St1),
	eval(E2,St,0),
	St2=St1.		% final value condition
eval(from(E1,S1,S2,E2),(Env,St),St2) :-			% LOOPMAIN	
	eval(E1,St,1),
	eval(S1,(Env,St),St1),
	eval(loop(E1,S1,E2,S2),(Env,St1),St2).
eval(loop(_E1,_S1,E2,_S2),(_Env,St),St1) :-			% LOOPBASE
	eval(E2,St,1),
	St1=St.			% final value condition
eval(loop(E1,S1,E2,S2),(Env,St),St3) :-			% LOOPREC
	eval(E2,St,0),
	eval(S2,(Env,St),St1),
	eval(E1,St1,0),
	eval(S1,(Env,St1),St2),
	eval(loop(E1,S1,E2,S2),(Env,St2),St3).
eval(skip,(_Env,St),St).								% SKIP
	
evalOp(add,X,Y,_St,V) :-
	V is X+Y.
evalOp(sub,X,Y,_St,V) :-
	V is X-Y.
evalOp(xor,X,Y,_St,V) :-
	V is X#Y.
	
find([(X,N)|_],X,N).
find([(Y,_)|St],X,N) :-
	X \== Y,
	find(St,X,N).

	
save(X,V,[(X,_)|St],[(X,W)|St]) :-
	W is V.
save(X,V,[(Y,M)|St],[(Y,M)|St1]) :-
	X \== Y,
	save(X,V,St,St1).
save(X,V,[],[(X,W)]) :-
	W is V.
	

gt(X,Y,1) :-
	X > Y.
gt(X,Y,0) :-
	X =< Y.
	
lt(X,Y,1) :-
	X < Y.
lt(X,Y,0) :-
	X >= Y.
	
gte(X,Y,1) :-
	X >= Y.
gte(X,Y,0) :-
	X < Y.
	
lte(X,Y,1) :-
	X =< Y.
lte(X,Y,0) :-
	X > Y.
	
eq(X,Y,1) :-
	X == Y.
eq(X,Y,0) :-
	X \== Y.
	
negate(true,false).
negate(false,true).

invertOp(add,sub).
invertOp(sub,add).

invert(asg(var(X),Op,E),asg(var(X),OpI,E)) :-
	invertOp(Op,OpI).
invert(asgarray(array(X,E1),Op,E2),asgarrary(array(X,E1),OpI,E2)) :-
	invertOp(Op,OpI).
invert(if(E1,S1,S2,E2),if(E2,S3,S4,E1)) :-
	invert(S1,S3),
	invert(S2,S4).
invert(from(E1,S1,S2,E2),from(E2,S3,S4,E1)) :-
	invert(S1,S3),
	invert(S2,S4).
invert(call(F),uncall(F)).
invert(uncall(F),call(F)).
invert(skip,skip).
invert(seq(S1,S2),seq(S3,S4)) :-
	invert(S1,S3),
	invert(S2,S4).

def(F,Env,S) :-
	member(proc(F,S),Env).
	
initStore([],[]).
initStore([X|Xs],[(X,0)|St]) :-
	initStore(Xs,St).
	
member(X,[X|_Xs]).
member(X,[Y|Xs]) :-
	member(X,Xs).

bigStepPred(eval(_,_,_)).

otherPred(evalOp(_,_,_,_,_)).
otherPred(find(_,_,_)).
otherPred(save(_,_,_,_)).
otherPred(gt(_,_,_)).
otherPred(lt(_,_,_)).
otherPred(gte(_,_,_)).
otherPred(lte(_,_,_)).
otherPred(eq(_,_,_)).
otherPred(negate(_,_)).
otherPred(def(_,_,_)).
otherPred(member(_,_)).
otherPred(invert(_,_)).


	
% Tests

% Example 1

ex1([proc(main,
	seq(asg(var(n),add,const(4)),
		seq(asg(var(x1),add,const(1)),
			seq(asg(var(x2),add,const(1)),
			call(fib))
		)
	)
	),
	proc(fib,
	from(
		var(x1)==var(x2),
		seq(asg(var(x1),add,var(x2)),
			seq(asg(var(x1),xor,var(x2)),
				seq(asg(var(x2),xor,var(x1)),
					asg(var(x1),xor,var(x2))
				)
			)
		),
		asg(var(n),sub,const(1)),
		var(n)==const(0)))
	]
	).

% Example 2

ex2a([proc(main,
	from(
		var(x1)==const(10),
		asg(var(x2),add,var(x1)),
		asg(var(x1),sub,const(1)),
		var(x1)==const(0)))
	]).
	
ex2b([proc(main,
	from(
		var(x1)==const(0),
		asg(var(x1),add,const(1)),
		asg(var(x2),add,var(x1)),
		var(x1)==const(0)))
	]).
	
	
% test

test(St) :-
	ex1(Env),
	initStore([x1,x2,n],St0),
	eval(call(main),(Env,St0),St).
\end{lstlisting}

\subsection{Derived small-step semantics for Janus }

The following Horn clauses are produced by specialising the contextual
small-step interpreter for big-step rules.

\begin{lstlisting}[basicstyle=\footnotesize\ttfamily]

% Clauses for a run
run__1(val(A),val(A)) :-
   true.
run__1(conf(A),B) :-
   smallStep__2(A,C),
   run__1(C,B).
   
 % Small step rules
smallStep__2((var(A),B),val(C)) :-
   eval__3(B,A,C).
smallStep__2((const(A),B),val(A)) :-
   true.
smallStep__2((skip,A,B),val(B)) :-
   true.
smallStep__2((call(A),B,C),D) :-
   eval__17(proc(A,E),B),
   smallStep__2((E,B,C),D).
smallStep__2((uncall(A),B,C),D) :-
   eval__17(proc(A,E),B),
   eval__5(E,F),
   smallStep__2((F,B,C),D).
smallStep__2((add(A,B),C),conf((add_3_2(D,B),C))) :-
   smallStep__2((A,C),val(D)).
smallStep__2((add(A,B),C),conf((add_3_1(B,D),C,E))) :-
   smallStep__2((A,C),conf((D,E))).
smallStep__2((sub(A,B),C),conf((sub_4_2(D,B),C))) :-
   smallStep__2((A,C),val(D)).
smallStep__2((sub(A,B),C),conf((sub_4_1(B,D),C,E))) :-
   smallStep__2((A,C),conf((D,E))).
smallStep__2((mul(A,B),C),conf((mul_5_2(D,B),C))) :-
   smallStep__2((A,C),val(D)).
smallStep__2((mul(A,B),C),conf((mul_5_1(B,D),C,E))) :-
   smallStep__2((A,C),conf((D,E))).
smallStep__2((div(A,B),C),conf((div_6_2(D,B),C))) :-
   smallStep__2((A,C),val(D)).
smallStep__2((div(A,B),C),conf((div_6_1(B,D),C,E))) :-
   smallStep__2((A,C),conf((D,E))).
smallStep__2((A>B,C),conf(('>_7_2'(D,B),C))) :-
   smallStep__2((A,C),val(D)).
smallStep__2((A>B,C),conf(('>_7_1'(B,D),C,E))) :-
   smallStep__2((A,C),conf((D,E))).
smallStep__2((A<B,C),conf(('<_8_2'(D,B),C))) :-
   smallStep__2((A,C),val(D)).
smallStep__2((A<B,C),conf(('<_8_1'(B,D),C,E))) :-
   smallStep__2((A,C),conf((D,E))).
smallStep__2((A>=B,C),conf(('>=_9_2'(D,B),C))) :-
   smallStep__2((A,C),val(D)).
smallStep__2((A>=B,C),conf(('>=_9_1'(B,D),C,E))) :-
   smallStep__2((A,C),conf((D,E))).
smallStep__2((A=<B,C),conf(('=<_10_2'(D,B),C))) :-
   smallStep__2((A,C),val(D)).
smallStep__2((A=<B,C),conf(('=<_10_1'(B,D),C,E))) :-
   smallStep__2((A,C),conf((D,E))).
smallStep__2((A==B,C),conf(('==_11_2'(D,B),C))) :-
   smallStep__2((A,C),val(D)).
smallStep__2((A==B,C),conf(('==_11_1'(B,D),C,E))) :-
   smallStep__2((A,C),conf((D,E))).
smallStep__2((logicaland(A,B),C),conf((logicaland_12_2(D,B),C))) :-
   smallStep__2((A,C),val(D)).
smallStep__2((logicaland(A,B),C),conf((logicaland_12_1(B,D),C,E))) :-
   smallStep__2((A,C),conf((D,E))).
smallStep__2((not(A),B),val(C)) :-
   smallStep__2((A,B),val(D)),
   eval__6(D,C).
smallStep__2((not(A),B),conf((not(C),D))) :-
   smallStep__2((A,B),conf((C,D))).
smallStep__2((asg(var(A),B,C),D,E),val(F)) :-
   eval__3(E,A,G),
   smallStep__2((C,E),val(H)),
   eval__7(B,G,H,E,I),
   eval__8(A,I,E,F).
smallStep__2((asg(var(A),B,C),D,E),conf((asg_14_1(A,F,B,G),E,H))) :-
   eval__3(E,A,F),
   smallStep__2((C,E),conf((G,H))).
smallStep__2((seq(A,B),C,D),conf((B,C,E))) :-
   smallStep__2((A,C,D),val(E)).
smallStep__2((seq(A,B),C,D),conf((seq_17_1(B,E),C,F))) :-
   smallStep__2((A,C,D),conf((E,F))).
smallStep__2((ifthenelse(A,B,C,D),E,F),conf((ifthenelse_18_2(D,B),F,E,F))) :-
   smallStep__2((A,F),val(1)).
smallStep__2((ifthenelse(A,B,C,D),E,F),conf((ifthenelse_18_1(B,D,G),E,F,H))) :-
   smallStep__2((A,F),conf((G,H))).
smallStep__2((ifthenelse(A,B,C,D),E,F),conf((ifthenelse_19_2(D,C),F,E,F))) :-
   smallStep__2((A,F),val(0)).
smallStep__2((ifthenelse(A,B,C,D),E,F),conf((ifthenelse_19_1(C,D,G),E,F,H))) :-
   smallStep__2((A,F),conf((G,H))).
smallStep__2((from(A,B,C,D),E,F),conf((from_20_2(A,B,C,D,B),E,E,F))) :-
   smallStep__2((A,F),val(1)).
smallStep__2((from(A,B,C,D),E,F),conf((from_20_1(A,B,C,D,G),E,F,H))) :-
   smallStep__2((A,F),conf((G,H))).
smallStep__2((loop(A,B,C,D),E,F),val(F)) :-
   smallStep__2((C,F),val(1)).
smallStep__2((loop(A,B,C,D),E,F),conf((loop_21_1(G),F,H))) :-
   smallStep__2((C,F),conf((G,H))).
smallStep__2((loop(A,B,C,D),E,F),conf((loop_22_2(C,A,B,D,D),E,E,F))) :-
   smallStep__2((C,F),val(0)).
smallStep__2((loop(A,B,C,D),E,F),conf((loop_22_1(A,B,C,D,G),E,F,H))) :-
   smallStep__2((C,F),conf((G,H))).
smallStep__2((add_3_1(A,B),C,D),conf((add_3_2(E,A),C))) :-
   smallStep__2((B,D),val(E)).
smallStep__2((add_3_1(A,B),C,D),conf((add_3_1(A,E),C,F))) :-
   smallStep__2((B,D),conf((E,F))).
smallStep__2((add_3_2(A,B),C),val(D)) :-
   smallStep__2((B,C),val(E)),
   D is A+E.
smallStep__2((add_3_2(A,B),C),conf((add_3_2(A,D),E))) :-
   smallStep__2((B,C),conf((D,E))).
smallStep__2((sub_4_1(A,B),C,D),conf((sub_4_2(E,A),C))) :-
   smallStep__2((B,D),val(E)).
smallStep__2((sub_4_1(A,B),C,D),conf((sub_4_1(A,E),C,F))) :-
   smallStep__2((B,D),conf((E,F))).
smallStep__2((sub_4_2(A,B),C),val(D)) :-
   smallStep__2((B,C),val(E)),
   D is A-E.
smallStep__2((sub_4_2(A,B),C),conf((sub_4_2(A,D),E))) :-
   smallStep__2((B,C),conf((D,E))).
smallStep__2((mul_5_1(A,B),C,D),conf((mul_5_2(E,A),C))) :-
   smallStep__2((B,D),val(E)).
smallStep__2((mul_5_1(A,B),C,D),conf((mul_5_1(A,E),C,F))) :-
   smallStep__2((B,D),conf((E,F))).
smallStep__2((mul_5_2(A,B),C),val(D)) :-
   smallStep__2((B,C),val(E)),
   D is A*E.
smallStep__2((mul_5_2(A,B),C),conf((mul_5_2(A,D),E))) :-
   smallStep__2((B,C),conf((D,E))).
smallStep__2((div_6_1(A,B),C,D),conf((div_6_2(E,A),C))) :-
   smallStep__2((B,D),val(E)).
smallStep__2((div_6_1(A,B),C,D),conf((div_6_1(A,E),C,F))) :-
   smallStep__2((B,D),conf((E,F))).
smallStep__2((div_6_2(A,B),C),val(D)) :-
   smallStep__2((B,C),val(E)),
   D is A/E.
smallStep__2((div_6_2(A,B),C),conf((div_6_2(A,D),E))) :-
   smallStep__2((B,C),conf((D,E))).
smallStep__2(('>_7_1'(A,B),C,D),conf(('>_7_2'(E,A),C))) :-
   smallStep__2((B,D),val(E)).
smallStep__2(('>_7_1'(A,B),C,D),conf(('>_7_1'(A,E),C,F))) :-
   smallStep__2((B,D),conf((E,F))).
smallStep__2(('>_7_2'(A,B),C),val(D)) :-
   smallStep__2((B,C),val(E)),
   eval__11(A,E,D).
smallStep__2(('>_7_2'(A,B),C),conf(('>_7_2'(A,D),E))) :-
   smallStep__2((B,C),conf((D,E))).
smallStep__2(('<_8_1'(A,B),C,D),conf(('<_8_2'(E,A),C))) :-
   smallStep__2((B,D),val(E)).
smallStep__2(('<_8_1'(A,B),C,D),conf(('<_8_1'(A,E),C,F))) :-
   smallStep__2((B,D),conf((E,F))).
smallStep__2(('<_8_2'(A,B),C),val(D)) :-
   smallStep__2((B,C),val(E)),
   eval__12(A,E,D).
smallStep__2(('<_8_2'(A,B),C),conf(('<_8_2'(A,D),E))) :-
   smallStep__2((B,C),conf((D,E))).
smallStep__2(('>=_9_1'(A,B),C,D),conf(('>=_9_2'(E,A),C))) :-
   smallStep__2((B,D),val(E)).
smallStep__2(('>=_9_1'(A,B),C,D),conf(('>=_9_1'(A,E),C,F))) :-
   smallStep__2((B,D),conf((E,F))).
smallStep__2(('>=_9_2'(A,B),C),val(D)) :-
   smallStep__2((B,C),val(E)),
   eval__13(A,E,D).
smallStep__2(('>=_9_2'(A,B),C),conf(('>=_9_2'(A,D),E))) :-
   smallStep__2((B,C),conf((D,E))).
smallStep__2(('=<_10_1'(A,B),C,D),conf(('=<_10_2'(E,A),C))) :-
   smallStep__2((B,D),val(E)).
smallStep__2(('=<_10_1'(A,B),C,D),conf(('=<_10_1'(A,E),C,F))) :-
   smallStep__2((B,D),conf((E,F))).
smallStep__2(('=<_10_2'(A,B),C),val(D)) :-
   smallStep__2((B,C),val(E)),
   eval__14(A,E,D).
smallStep__2(('=<_10_2'(A,B),C),conf(('=<_10_2'(A,D),E))) :-
   smallStep__2((B,C),conf((D,E))).
smallStep__2(('==_11_1'(A,B),C,D),conf(('==_11_2'(E,A),C))) :-
   smallStep__2((B,D),val(E)).
smallStep__2(('==_11_1'(A,B),C,D),conf(('==_11_1'(A,E),C,F))) :-
   smallStep__2((B,D),conf((E,F))).
smallStep__2(('==_11_2'(A,B),C),val(D)) :-
   smallStep__2((B,C),val(E)),
   eval__15(A,E,D).
smallStep__2(('==_11_2'(A,B),C),conf(('==_11_2'(A,D),E))) :-
   smallStep__2((B,C),conf((D,E))).
smallStep__2((logicaland_12_1(A,B),C,D),conf((logicaland_12_2(E,A),C))) :-
   smallStep__2((B,D),val(E)).
smallStep__2((logicaland_12_1(A,B),C,D),conf((logicaland_12_1(A,E),C,F))) :-
   smallStep__2((B,D),conf((E,F))).
smallStep__2((logicaland_12_2(A,B),C),val(D)) :-
   smallStep__2((B,C),val(E)),
   D is A/\E.
smallStep__2((logicaland_12_2(A,B),C),conf((logicaland_12_2(A,D),E))) :-
   smallStep__2((B,C),conf((D,E))).
smallStep__2((asg_14_1(A,B,C,D),E,F),val(G)) :-
   smallStep__2((D,F),val(H)),
   eval__7(C,B,H,E,I),
   eval__8(A,I,E,G).
smallStep__2((asg_14_1(A,B,C,D),E,F),conf((asg_14_1(A,B,C,G),E,H))) :-
   smallStep__2((D,F),conf((G,H))).
smallStep__2((seq_17_1(A,B),C,D),conf((A,C,E))) :-
   smallStep__2((B,D),val(E)).
smallStep__2((seq_17_1(A,B),C,D),conf((seq_17_1(A,E),C,F))) :-
   smallStep__2((B,D),conf((E,F))).
smallStep__2((ifthenelse_18_1(A,B,C),D,E,F),conf((ifthenelse_18_2(B,A),E,D,E))) :-
   smallStep__2((C,F),val(1)).
smallStep__2((ifthenelse_18_1(A,B,C),D,E,F),conf((ifthenelse_18_1(A,B,G),D,E,H))) :-
   smallStep__2((C,F),conf((G,H))).
smallStep__2((ifthenelse_18_2(A,B),C,D),conf((ifthenelse_18_3(E,A),C))) :-
   smallStep__2((B,D),val(E)).
smallStep__2((ifthenelse_18_2(A,B),C,D),conf((ifthenelse_18_2(A,E),C,F))) :-
   smallStep__2((B,D),conf((E,F))).
smallStep__2((ifthenelse_18_3(A,B),C),val(A)) :-
   smallStep__2((B,C),val(1)).
smallStep__2((ifthenelse_18_3(A,B),C),conf((ifthenelse_18_3(A,D),E))) :-
   smallStep__2((B,C),conf((D,E))).
smallStep__2((ifthenelse_19_1(A,B,C),D,E,F),conf((ifthenelse_19_2(B,A),E,D,E))) :-
   smallStep__2((C,F),val(0)).
smallStep__2((ifthenelse_19_1(A,B,C),D,E,F),conf((ifthenelse_19_1(A,B,G),D,E,H))) :-
   smallStep__2((C,F),conf((G,H))).
smallStep__2((ifthenelse_19_2(A,B),C,D),conf((ifthenelse_19_3(E,A),C))) :-
   smallStep__2((B,D),val(E)).
smallStep__2((ifthenelse_19_2(A,B),C,D),conf((ifthenelse_19_2(A,E),C,F))) :-
   smallStep__2((B,D),conf((E,F))).
smallStep__2((ifthenelse_19_3(A,B),C),val(A)) :-
   smallStep__2((B,C),val(0)).
smallStep__2((ifthenelse_19_3(A,B),C),conf((ifthenelse_19_3(A,D),E))) :-
   smallStep__2((B,C),conf((D,E))).
smallStep__2((from_20_1(A,B,C,D,E),F,G,H),conf((from_20_2(A,B,C,D,B),F,F,G))) :-
   smallStep__2((E,H),val(1)).
smallStep__2((from_20_1(A,B,C,D,E),F,G,H),conf((from_20_1(A,B,C,D,I),F,G,J))) :-
   smallStep__2((E,H),conf((I,J))).
smallStep__2((from_20_2(A,B,C,D,E),F,G),conf((loop(A,B,D,C),F,H))) :-
   smallStep__2((E,G),val(H)).
smallStep__2((from_20_2(A,B,C,D,E),F,G),conf((from_20_2(A,B,C,D,H),F,I))) :-
   smallStep__2((E,G),conf((H,I))).
smallStep__2((loop_21_1(A),B,C),val(B)) :-
   smallStep__2((A,C),val(1)).
smallStep__2((loop_21_1(A),B,C),conf((loop_21_1(D),B,E))) :-
   smallStep__2((A,C),conf((D,E))).
smallStep__2((loop_22_1(A,B,C,D,E),F,G,H),conf((loop_22_2(C,A,B,D,D),F,F,G))) :-
   smallStep__2((E,H),val(0)).
smallStep__2((loop_22_1(A,B,C,D,E),F,G,H),conf((loop_22_1(A,B,C,D,I),F,G,J))) :-
   smallStep__2((E,H),conf((I,J))).
smallStep__2((loop_22_2(A,B,C,D,E),F,G),conf((loop_22_3(D,A,B,C,B),F,H,H))) :-
   smallStep__2((E,G),val(H)).
smallStep__2((loop_22_2(A,B,C,D,E),F,G),conf((loop_22_2(A,B,C,D,H),F,I))) :-
   smallStep__2((E,G),conf((H,I))).
smallStep__2((loop_22_3(A,B,C,D,E),F,G,H),conf((loop_22_4(C,A,B,D,D),F,F,G))) :-
   smallStep__2((E,H),val(0)).
smallStep__2((loop_22_3(A,B,C,D,E),F,G,H),conf((loop_22_3(A,B,C,D,I),F,G,J))) :-
   smallStep__2((E,H),conf((I,J))).
smallStep__2((loop_22_4(A,B,C,D,E),F,G),conf((loop(A,D,C,B),F,H))) :-
   smallStep__2((E,G),val(H)).
smallStep__2((loop_22_4(A,B,C,D,E),F,G),conf((loop_22_4(A,B,C,D,H),F,I))) :-
   smallStep__2((E,G),conf((H,I))).
   
% eval__3(St,X,V) = find(St,X,V)

eval__3([(A,B)|C],A,B) :-
   true.
eval__3([(A,B)|C],D,E) :-
   D\==A,
   eval__3(C,D,E).
   
% eval__5(X,Y) = invert(X,Y)

eval__5(asg(var(A),B,C),asg(var(A),D,C)) :-
   eval__18(B,D).
eval__5(asgarray(array(A,B),C,D),asgarrary(array(A,B),E,D)) :-
   eval__18(C,E).
eval__5(if(A,B,C,D),if(D,E,F,A)) :-
   eval__5(B,E),
   eval__5(C,F).
eval__5(from(A,B,C,D),from(D,E,F,A)) :-
   eval__5(B,E),
   eval__5(C,F).
eval__5(call(A),uncall(A)) :-
   true.
eval__5(uncall(A),call(A)) :-
   true.
eval__5(skip,skip) :-
   true.
eval__5(seq(A,B),seq(C,D)) :-
   eval__5(A,C),
   eval__5(B,D).
   
% eval__6(X,Y) = negate(X,Y)

eval__6(true,false) :-
   true.
eval__6(false,true) :-
   true.
   
% arithmetic operations
eval__7(add,A,B,C,D) :-
   D is A+B.
eval__7(sub,A,B,C,D) :-
   D is A-B.
eval__7(xor,A,B,C,D) :-
   D is A#B.
eval__8(A,B,[(A,C)|D],[(A,E)|D]) :-
   E is B.
eval__8(A,B,[(C,D)|E],[(C,D)|F]) :-
   A\==C,
   eval__8(A,B,E,F).
eval__8(A,B,[],[(A,C)]) :-
   C is B.
eval__11(A,B,1) :-
   A>B.
eval__11(A,B,0) :-
   A=<B.
eval__12(A,B,1) :-
   A<B.
eval__12(A,B,0) :-
   A>=B.
eval__13(A,B,1) :-
   A>=B.
eval__13(A,B,0) :-
   A<B.
eval__14(A,B,1) :-
   A=<B.
eval__14(A,B,0) :-
   A>B.
eval__15(A,B,1) :-
   A==B.
eval__15(A,B,0) :-
   A\==B.
eval__17(A,[A|B]) :-
   true.
eval__17(A,[B|C]) :-
   eval__17(A,C).
eval__18(add,sub) :-
   true.
eval__18(sub,add) :-
   true.

\end{lstlisting}

\section{The Horn clause interpreter for big-step semantic rules}\label{secA3}

The interpreter used in the specialisation experiments is listed below.  
It applies contextual small-step semantics to big-step semantic rules, incorporating
folding of conjunctions, and the optional refinements described in Section \ref{optional}.

The suffix rules are generated before interpretation, by the call \texttt{auxRules}.

\begin{lstlisting}[basicstyle=\footnotesize\ttfamily]
:- module(ss_final,_).

:- use_module(chclibs(readprog)).
:- use_module(chclibs(common)).
:- use_module(instance).
:- use_module(bigStepSuffixRules).

go(F,A) :-
	readprog(F,P),
	auxRules(P,PAux),
	arrowArgs(A,S,V,P),
	run(conf(S),val(V),PAux).
	

run(val(V),val(V),_).
run(conf(S),C,P) :-
	smallStep(conf(S),C1,P),
	run(C1,C,P).
	
smallStep(conf(S),val(V),P) :-		% unit clause Rule [UNF0]
	arrowRule(_K,S,V,R,P),
	evalCondition(R,V,[],P).
smallStep(conf(S),C1,P) :-			% non-unit clause, singleton body Rule [UNF1]
	arrowRule(_K,S,V,R,P),
	evalCondition(R,V,[B],P),
	arrowArgs(B,S1,V,P),
	smallStep(conf(S1),C1,P).
smallStep(conf(S),C1,P) :-			% non-unit clause, doubleton body
	arrowRule(K,S,V,R,P),
	evalCondition(R,V,[B1|Bs],P),
	Bs\==[],
	arrowArgs(B1,S1,V1,P),
	smallStep(conf(S1),C3,P),
	nextStep(K,C3,V1,V,Bs,C1,P).

nextStep(_K,val(V1),V1,V,Bs,val(V),P) :-	% Rule [FOLD0_1]
	evalCondition(Bs,V,[],P).
nextStep(K,val(V1),V1,V,Bs,conf(S),P) :-	% Rule [FOLD0_2]
	evalCondition(Bs,V,Bs1,P),
	Bs1\==[],
	fold(K,A1,Bs1,P),
	arrowArgs(A1,S,V,P).
nextStep(K,conf(S1),V1,V,Bs,conf(S),P) :-	% Rule [FOLD1]
	arrowArgs(B3,S1,V1,P),
	fold(K,A1,[B3|Bs],P),
	arrowArgs(A1,S,V,P).
	
fold(_K,true,[],_).
fold(_K,B,[B],_).
fold(K-_,A,Bs,P) :-
	Bs=[_,_|_],
	member(auxcls(Cls,K),P),
	member(clause(Cl,K-_),Cls),	
	instanceOf(Cl,(A:-Bs)).
fold(K,A,Bs,P) :-
	Bs=[_,_|_],
	number(K),
	member(auxcls(Cls,K),P),
	member(clause(Cl,K-_),Cls),	
	instanceOf(Cl,(A:-Bs)).

%======================== 

arrowRule(K,S,V,R,P) :-
	arrowArgs(A,S,V,P),
	rule(K,A,R,P).
	

rule(K,A,Bs,P) :-
	member(clause(Cl,K),P),		
	instanceOf(Cl,(A:-Bs)).
rule(K,A,Bs,P) :-
	number(K),
	member(auxcls(Cls,K),P),
	member(clause(Cl,K-_),Cls),
	instanceOf(Cl,(A:-Bs)).
rule(K-J,A,Bs,P) :-
	member(auxcls(Cls,K),P),
	member(clause(Cl,K-J),Cls),
	J>0,		
	instanceOf(Cl,(A:-Bs)).
	
%========================
		
evalCondition([],_,[],_).
evalCondition([B1|R],_,[B1|R],P) :-
	bigStepPred(B1,P).
evalCondition([B|Bs],V,Bs1,P) :-
	otherPred(B,P),
	eval(B,P),
	evalCondition(Bs,V,Bs1,P).

eval(B,_) :-
	constraint(B),
	call(B).
eval(B,P) :-
	rule(_,B,B1,P),
	callPreds(B1,P).
	
callPreds([],_).
callPreds([B|Bs],P) :-
	eval(B,P),
	callPreds(Bs,P).
	
arrowArgs(A,Conf,V,P) :-
	bigStepPred(A,P),
	A=..[_|Args],
	append(S,[V],Args),
	list2Conj(S,Conf).	 % Final arg is value
	
bigStepPred(A,P) :- 
	rule(_,bigStepPred(A),_,P).

otherPred(A,P) :- 
	rule(_,otherPred(A),_,P).
otherPred(B,_) :-
	constraint(B).

constraint(_ = _).
constraint(_ < _).
constraint(_ > _).
constraint(_ =< _).
constraint(_ >= _).
constraint(_ is _).
constraint(_ =:= _).
constraint(_\==_).
constraint(_==_).
constraint(_\=_).


\end{lstlisting}

\end{appendices}

\end{document}